\documentclass[a4paper]{article}
\usepackage[a4paper,top=3.25cm,left=2.5cm,right=2.75cm,bottom=3cm]{geometry}
\usepackage[english]{babel}
\usepackage{amssymb,amsmath,amsthm,graphicx}
\usepackage[round,authoryear]{natbib}
\usepackage{enumitem}
\usepackage{setspace}
\usepackage[table]{xcolor}
\usepackage{multirow}
\usepackage{longtable}
\usepackage{graphicx}
\usepackage{caption}
\usepackage{graphics}
\usepackage{natbib}
\usepackage[breaklinks=true]{hyperref}
\usepackage{breakcites}
\usepackage{mathtools}
\usepackage{tikz}
\usepackage{standalone}
\usetikzlibrary{arrows,petri,topaths,automata,shapes}
\usepackage[linesnumbered,ruled]{algorithm2e}
\usepackage{setspace}
\usepackage[section]{placeins}
\usepackage{hyperref}
\usepackage{url}
\usepackage{breakurl}
\hypersetup{colorlinks,%
citecolor=black,%
filecolor=blue,%
linkcolor=blue,%
urlcolor=blue}%
\usepackage{booktabs}
\usepackage{microtype}
\usepackage[toc,page]{appendix}
\usepackage{lscape}
\usepackage{fancyhdr}
\newcommand{\uppp}[1]{\raisebox{2ex}[0pt]{#1}}

\newcolumntype{C}[1]{>{\centering\let\newline\\\arraybackslash\hspace{0pt}}m{#1}}
\newcolumntype{H}{>{\setbox0=\hbox\bgroup}c<{\egroup}@{}}
\newtheorem{proposition}{Proposition}

\begin{document}
\large

\title{Mathematical models and search algorithms for the capacitated $p$-center problem}

\author{\textbf{Raphael Kramer, Manuel Iori} \\
Dipartimento di Scienze e Metodi dell'Ingegneria \\ Universit{\`a} degli Studi di Modena e Reggio Emilia, Italy \\ 
\{raphael.kramer, manuel.iori\}@unimore.it
\and
\textbf{Thibaut Vidal} \\
Departamento de Inform{\'a}tica \\ Pontif{\'i}cia Universidade Cat{\'o}lica do Rio de Janeiro, Brazil \\
vidalt@inf.puc-rio.br§
}
\date{}

\maketitle

\vspace{-0.5cm}
\begin{center}
Technical Report -- March 2018 \\
\end{center}
\vspace{0.5cm}

\begin{abstract}

The capacitated $p$-center problem requires to select $p$ facilities from a set of candidates to service a number of customers, subject to facility capacity constraints, with the aim of minimizing the maximum distance between a customer and its associated facility. The problem is well known in the field of facility location, because of the many applications that it can model. In this paper, we solve it by means of search algorithms that iteratively seek the optimal distance by solving tailored subproblems. 
We present different mathematical formulations for the subproblems and improve them by means {of several valid inequalities, including an effective one based on a 0-1 disjunction and the solution of subset sum problems.
We also develop an alternative search strategy that finds a balance between the traditional sequential search and binary search. This strategy limits the number of feasible subproblems to be solved and, at the same time, avoids large overestimates of the solution value, which are detrimental for the search. We evaluate the proposed techniques by means of extensive computational experiments on benchmark instances from the literature and new larger test sets. All instances from the literature with up to 402 vertices and integer distances are solved to proven optimality, including 13 open cases, and feasible solutions are found in 10 minutes for instances with up to 3038 vertices. }

\end{abstract}

\onehalfspace

\section{Introduction} 

Facility location problems play an important role in the operations research and combinatorial optimization literature. Since the early years of these disciplines, considerable research effort has been dedicated to their solution, due to their importance in supply chain management, healthcare, telecommunications, humanitarian relief, and machine learning applications, among many others (see, e.g., \citealt{DreznerHamacher2002} and \citealt{Laporte2015a}).

The \textit{$p$-center problem} (PCP) is a particular facility location problem that requires to select $p$ facilities, from a set of candidates, to serve a set of customers in such a way that the maximum distance between a customer and its closest facility is minimized (see, e.g., \citealt{Hakimi1964} and \citealt{Minieka1970}). Two problem variants are usually distinguished: the \emph{absolute} PCP, where the facilities can be located on the vertices or on the edges of a graph, and the \emph{vertex} PCP, considered here, in which the facilities can be located only on the vertices (see, e.g., \citealt{KarivHakimi1979}).
The main applications of the PCP arise in the location of emergency facilities, such as fire stations, police stations, and ambulance waiting locations, in a context where the worst-case service level, e.g., 
service time or distance, must be bounded or minimized (see, e.g., \citealt{Daskin1995ch5} and \citealt{MarianovRevelle1995}).
The PCP is known to be $\mathcal{NP}$-hard \citep{KarivHakimi1979,Masuyama1981}, and thus several heuristics have been proposed for its solution, making use of paradigms such as tabu search \citep{Mladenovicetal2003}, bee colony optimization \citep{Davidovic2011}, variable neighborhood search \citep{Irawan2015}, and other techniques (\citealt{Hochbaum1985}, \citealt{Plesnik1987}, \citealt{Mihelic2003}, and \citealt{Davoodi2011}). From the perspective of exact PCP methods, the most successful type of approach, to the best of our knowledge, consists in solving a series of covering subproblems with the help of preprocessing and reduction techniques (\citealt{Daskin2000}, \citealt{Ilhanetal2002}, \citealt{Elloumi2004}, \citealt{Chen2009}, and \citealt{CalikTansel2013}).

In this paper, we focus on the capacitated version of the vertex PCP, in which each customer is characterized by a demand and each candidate facility by a capacity. Each demand must be serviced integrally by one of the $p$ chosen facilities, without exceeding their capacities. This feature renders the problem more challenging: a customer may no longer be assigned to its closest facility, and thus the complete characterization of a solution requires both facility opening and customer-facility allocation decisions.

To solve the problem, we propose a decomposition algorithm that searches for the optimal distance by iteratively solving a capacitated set covering formulation.
The decomposition is enhanced by devising tailored search algorithms as well as valid inequalities and symmetry breaking rules for the set covering subproblems. 
We also introduce an alternative arc-flow formulation, which is shown to provide good linear-relaxation bounds. Computational experiments conducted on a large set of benchmark instances show that our techniques help to find and prove optimal solutions. In particular, we provide the optimal solutions for all the benchmark instances with up to 402 vertices considered in the literature, including 13 open cases. Finally, to better evaluate the performance of the proposed algorithms, we introduce and solve new instances containing between 50 and 3000 vertices.

The remainder of this paper is organized as follows. 
Section \ref{sec:descrip} formally describes the problem and reviews the related literature.
Section \ref{sec:formulations} presents the mathematical formulations and the decomposition-based approach.
Section \ref{sec:improvements} introduces improvement methods, whereas Section~\ref{sec:method} presents classical and alternative decomposition-based search algorithms. Section \ref{sec:results} reports our computational experiments and Section~\ref{sec:conclusion} concludes.

\FloatBarrier
\section{Problem Description and Literature Review} \label{sec:descrip}

The \textit{capacitated $p$-center problem} (CPCP) is defined as follows. Let $G=(F, C, E)$ be a bipartite graph, where $F=\{1,\dots,m\}$ and $C=\{1,\dots,n\}$ are the sets of candidate locations and customer nodes, respectively, and $E = \{(i,j) : i \in F, j \in C\}$ is the set of edges. Each edge $(i,j) \in E$ represents a possible assignment of a customer $j$ to a facility $i$ and has a non-negative distance $d_{ij}$.
{Each facility $i \in F$ has a capacity $Q_i$, which can be used to supply the demands $q_j$ of customers $j \in C$. Unless stated otherwise, we assume in the following that the input parameters are integer.}
The CPCP aims at opening at most $p$ facilities and assigning each customer to exactly one facility, seeking to minimize the maximum distance between a customer and its facility, and ensuring that the total demand of the customers assigned to each facility does not exceed its capacity.

The CPCP is closely related to the \textit{capacitated p-median problem} (CPMP), the latter differing only in the objective function, which minimizes the \emph{sum} of all the customer-facility assignments. The literature on the CPMP is abundant, especially for heuristic methods (see, e.g., \citealp{Reese2006}). In terms of exact algorithms, we highlight the column generation approach of \citet{LorenaSenne2004} and the branch-and-price algorithm of \citet{Ceselli2005}. In the former article, the authors formulate the restricted master problem as a set covering model, solve $m$ binary knapsack subproblems to generate the columns, and use a Lagrangian/surrogate relaxation to speed up the convergence. In the latter article, the authors use a similar column generation approach to obtain the linear relaxation for the nodes of a branch-and-bound algorithm. {Branching is performed first on location variables, and then on assignment variables. Their best algorithm} has good performance on small instances and on instances with a large {$n/p$} ratio. The benchmark sets introduced in these articles are now commonly used to evaluate solution procedures for the CPCP. In addition to these works, \cite{Boccia2008} proposed a cutting plane algorithm based on Fenchel cuts (see \citealt{Boyd1993,Boyd1994}), which improved the gaps provided by the previously cited works and led to new proven optimal solutions. For what concerns exact approaches for the uncapacitated version of the $p$-median problem, we point out the column-and-row generation method by \citet{Garcia2011}, and the branch decomposition algorithm by \citet{Fast2017}.

In contrast with the CPMP, the literature on the CPCP is limited. Concerning exact methods, \cite{Jaeger1994} presented a polynomial algorithm for the special case in which assignments among customers and facilities are organized on a tree network and capacities are identical. {\citet{Ozsoy2006} proposed binary search algorithms that iteratively solve either a capacitated concentrator location problem or a variant of the bin packing problem. Their computational experiments demonstrate that these search algorithms provide better results than the solution of a compact formulation of the problem.}
The idea of decomposing the original problem into smaller subproblems was also adopted by \citet{AlbaredaSambola2010}. They considered two subproblems: a capacitated concentrator location problem, and a capacitated maximal covering problem. Moreover, instead of directly solving the subproblem formulations via available \textit{mixed integer linear programming} (MILP) solvers, they performed a Lagrangian relaxation of the assignment constraints to improve the lower bounds, and they applied heuristic procedures in the inner iterations of a subgradient optimization to generate feasible solutions. The resulting algorithm had good convergence and improved the results of \citet{Ozsoy2006} in most cases.

Concerning heuristic approaches, \citet{Scaparraetal2004} proposed a very large neighborhood search, relying on flow-based algorithms to efficiently detect improving neighbors, while \citet{QuevedoOrozco2015} addressed the problem with an iterated greedy local search with variable neighborhood descent. We also mention the work of \cite{Espejo2015}, who proposed a CPCP variant in which the maximum distance to the second-closest center is minimized. The problem is relevant in situations where the facilities can become unavailable because of unforeseen events such as natural disaster or a labor strike. Mathematical formulations, heuristics, and preprocessing procedures were proposed and experimentally evaluated.

\FloatBarrier
\section{Mathematical Formulations and Decomposition Approach} \label{sec:formulations}

In this section, we provide two formulations for the CPCP and then describe the relationship between the CPCP and the capacitated set covering problem, which is the foundation of our search algorithms. 
We need some additional notation. Let $D = [d_{ij}]$ define the cost matrix and $r \in D$ denote a \emph{coverage radius}. In our search algorithms, $r$ represents the maximum distance allowed for a customer-facility assignment. Moreover, the circular area having radius $r$ and center in facility $i$ is called the \emph{coverage area} of $i$.
Let $C^r_i = \{j \in C: d_{ij} \leq r \text{ and } q_j \leq Q_i\}$ be the set of customers {lying in the coverage area of facility~$i$ induced by $r$}, and $F^r_j = \{i \in F: d_{ij} \leq r \text{ and } q_j \leq Q_i\}$ be the set of facilities that can cover customer $j$ within distance $r$.

\FloatBarrier
\subsection{Descriptive formulation} \label{sec:desc-formulation}

{Let $y_i$ be a binary variable taking value 1 if facility $i \in F$ is open and 0 otherwise, and $x_{ij}$ a binary variable stating whether or not customer $j \in C$ is assigned to facility $i \in F$. Let $z$ be a non-negative variable that keeps track of the maximum distance over all customer-facility assignments. The CPCP can be modeled with the following {\em descriptive} formulation:}
\allowdisplaybreaks
\begin{align}
\text{(CPCP-D) } \min z \label{CpCP:FO} \\
\text{s.t.}  \sum_{i \in F} x_{ij} = 1 & \qquad j \in C, \label{CpCP:const1} \\
 \sum_{i \in F} y_{i} \leq p & \qquad \label{CpCP:const3} \\
 z \geq \sum_{i \in F} d_{ij} x_{ij} & \qquad j \in C, \label{CpCP:const4} \\
 \sum_{j \in C} q_{j} x_{ij} \leq Q_{i} y_{i} & \qquad i \in F, \label{CpCP:const5} \\
 x_{ij} \in \{0,1\} & \qquad i \in F, j \in C, \label{CpCP:const6} \\
 y_{i} \in \{0,1\} & \qquad i \in F. \label{CpCP:const7}
\end{align}

Objective function \eqref{CpCP:FO} minimizes the maximum distance. Constraints \eqref{CpCP:const1} ensure that each customer is assigned to one facility.
The total number of open facilities is limited to $p$ by constraint~\eqref{CpCP:const3}.
Constraints \eqref{CpCP:const4} force $z$ to be greater than or equal to the distance from any customer to its assigned facility.
Constraints \eqref{CpCP:const5} ensure that the sum of demands assigned to an open facility does not exceed its capacity, and 
constraints \eqref{CpCP:const6} and \eqref{CpCP:const7} provide the binary conditions. 

\FloatBarrier
\subsection{Extended arc-flow formulation} \label{sec:arc-flow}

Arc-flow formulations model a problem by using a capacitated network of pseudo-poly\-no\-mial size. Successful {formulations of this type have been presented for a number of combinatorial optimization problems (see, e.g., the works on bin packing and cutting stock by \citealt{ValeriodeCarvalho2002} and \citealt{DIM15})}, but, to the best of our knowledge, this is the first time they have been applied to  the CPCP. To construct a CPCP arc-flow formulation, first of all we associate {with each $i \in F$ an acyclic directed multigraph $G_i = (V_i, A_i)$, with $V_i=\{0,1,\dots,Q_i\}$ and $A_i=\{(e,f,j): e, f \in V_i \mbox{ and }j \in C \cup \{0\}\}$}. 
The vertices represent partial fillings of the facility capacity, whereas the arcs have a double meaning: {\em customer arcs} represent the assignment of a customer {to a facility}, whereas  {\em loss arcs} represent the unused residual capacity of the facility. Formally, for any $i \in F$ we partition $A_i = \cup_{j \in C \cup \{0\}} A_{ij}$, where $A_{i0} = \{(d,Q_i,0): d \in V_i \setminus Q_i\}$ is the set of loss arcs, and $A_{ij} = \{(d,d+q_j,j): d, d+q_j \in V_i \}$ are the sets of customer arcs for all $j \in C$.

A valid assignment of customers to {a facility} corresponds to a path $\mathcal{P} \subseteq G_i$ containing one or more customer arcs and at most one loss arc. If arc $(d,e,j) \in \mathcal{P}$, then either customer $j \in C$ (with demand $q_j = e-d$) is served by $i$, or a residual capacity of $e-d$ $(=Q_i-d)$ units is unused. The aim of the arc-flow formulation is to open facilities and assign a path to each open facility by minimizing the maximum assignment distance and ensuring that all customers are served. 
Let $A=\cup_{i \in F} A_i$ be the set of all arcs, $a$ be the index of a generic arc in $A$, $\delta_{ij}^{+}(e)$ be the subset of arcs $a \in A_{ij}$ exiting from node $e$, and $\delta_{ij}^{-}(e)$ be the subset of arcs $a \in A_{ij}$ entering node $e$. Let $f_{aij}$ be a binary variable that takes value 1 if arc $a \in A_{ij}$ is selected and 0 otherwise. The CPCP can be modeled as
\allowdisplaybreaks
\begin{align}
\text{(CPCP-AF) } \min z \nonumber \\
\text{s.t. \eqref{CpCP:const1}--\eqref{CpCP:const4}, \eqref{CpCP:const6}, \eqref{CpCP:const7}, and} \nonumber \\
 \sum_{j \in C} \sum_{a \in \delta_{ij}^{-}(e)} f_{aij} - 
       \sum_{j \in C} \sum_{a \in \delta_{ij}^{+}(e)} f_{aij} = \left\{
{\begin{tabular}{@{\hspace{0.0cm}}r@{\hspace{0.01cm}}l@{\hspace{0.01cm}}}
\vspace{0.2cm}$-y_i$  & $\text{ if } e = 0$ \\
\vspace{0.2cm}$y_i$ & $ \text{ if } e = Q_i$ \\
\vspace{0.0cm}$0$    & $ \text{ otherwise}$
\end{tabular}} \right. & \qquad i \in F, e \in V_i, \label{eq:AFinit}\\
 \sum_{e \in V_i} \sum_{a \in \delta_{ij}^{+}(e)} f_{aij} = x_{ij} & \qquad i \in F, j \in C, \label{eq:AF2}\\ 
 f_{aij} \in \{0,1\} & \qquad i \in F, j \in C, a \in A_{ij}. \label{eq:AFend}
\end{align}

Constraints \eqref{eq:AFinit} impose the flow conservation at the nodes $V_i$ for all the facilities. Constraints \eqref{eq:AF2} ensure that the number of arcs associated with customer $j$ in facility $i$ is equal to  $x_{ij}$ (thus being either 1 or 0). Note that CPCP-AF is an extended formulation of CPCP-D, and practical CPCP-AF solutions may be obtained by replacing $x_{ij}$ with $\sum_{e \in V_i} \sum_{a \in \delta_{ij}^{+}(e)} f_{aij}$ in constraints \eqref{CpCP:const1} and \eqref{CpCP:const4} and removing constraints \eqref{CpCP:const6} and \eqref{eq:AF2}.

This formulation provides good lower bounds. However, due to its pseudo-polynomial number of variables ($\mathcal{O}(Q_i|C|)$ for each $i$), it becomes impracticable for large instances. Hence, we introduce reduction techniques. First, one can use an upper bound $r$, i.e., a valid CPCP coverage radius, and remove the assignment variables associated with distances greater than $r$.
Secondly, it is important to know which customers are served by a facility, but not their sequence. Thus, one can sort the customers according to an arbitrary criterion and so reduce the number of arcs. 
Similarly to \cite{ValeriodeCarvalho2002}, for each facility $i$, we sort the customers in $C^r_i$ by non-increasing $q_j$ value. Then, we build the arcs by considering this order (arcs associated with the first customer can only start at 0, arcs associated with the second customer can start either at 0 or right after the arcs of the first customer, and so on). This was done by a standard dynamic programming algorithm.

\FloatBarrier
\subsection{Decomposition Approach}\label{subsec:decomposition}

The objective function of the CPCP minimizes the maximum distance among the selected customer-facility assignments. The optimal solution value is therefore contained in the distance matrix $D$, and one can test with a search algorithm whether a solution value~{$r$} is optimal. More precisely, testing whether the CPCP admits a solution with optimal value $z \leq r$ is equivalent to solving a feasibility test where all the capacity constraints are satisfied and only customer-facility assignments of distance {$d_{ij} \leq r$} are used. 
On the basis of preliminary experimental analyses, we decided to transform this feasibility test into the problem of minimizing the number of selected facilities that cover all the customers while satisfying the capacity constraints. Formally, this corresponds to the following \textit{capacitated set covering problem} (CSCP-$r$):
\allowdisplaybreaks
\begin{align}
\text{(CSCP-$r$) } \min \sum_{i \in F} y_i \label{eq:CSCPinit} \\
\text{s.t.} \sum_{i \in F_j^r} x_{ij} \geq 1 & \qquad j \in C, \label{eq:CSCP2} \\
 \sum_{j \in C_i^r} q_j x_{ij} \leq Q_i y_i & \qquad i \in F, \label{eq:CSCP5} \\
 x_{ij} \in \{0,1\} & \qquad i \in F, j \in C_i^r, \label{eq:CSCP6} \\
 y_{i} \in \{0,1\} & \qquad i \in F. \label{eq:CSCPend}
\end{align}

Objective function \eqref{eq:CSCPinit} minimizes the number of selected facilities. Constraints~\eqref{eq:CSCP2} ensure that each customer $j \in C$ is assigned to at least one facility, whereas constraints~\eqref{eq:CSCP5} impose capacity restrictions on all the facilities. Note that in constraints \eqref{eq:CSCP2}, we opted to replace the ``$=$'' (originally used in constraints \eqref{CpCP:const1}) with a ``$\geq$''. {This can be done without loss of optimality and allows us to develop improvement methods (given in Section~\ref{sec:improvements} below)}. Clearly, a solution in which a customer is entirely assigned to more than one facility can be mapped to a solution with the same cost, in which such the customer is assigned to a single facility. Note that an extended arc-flow formulation for CSCP-$r$, denoted CSCP-AF-$r$, can also be obtained by replacing constraints \eqref{eq:CSCP5} with \eqref{eq:AFinit}--\eqref{eq:AFend} and removing the assignment variables associated with distances $d_{ij} > r$.

Based on this decomposition, the success of the approach now depends on two factors: 1) an efficient solution of each CSCP-$r$, and 2) an efficient search for the optimal value $r \in D$, with the smallest total computational effort. {Note that the smallest computational effort is not necessarily proportional to the number of calls to the CSCP-$r$ subproblems, since some subproblems are simpler to solve (e.g., when $r$ is small) than others. In fact, there can be large differences in the computational effort required to either find a feasible solution or prove that no feasible solution exists for a given radius}. To address these two challenges, the next two sections present techniques that improve the solution of the subproblems (Section \ref{sec:improvements}) and search strategies to find the optimal coverage radius (Section \ref{sec:method}).

\FloatBarrier
\section{Improvement Methods for the CSCP-$r$} \label{sec:improvements}

To ease the description, let $\Delta_{i}(r) = \sum_{j \in C_{i}^r} q_j$ be the sum of the demands of the customers located inside the {coverage area} defined by facility $i$ and {coverage radius} $r$, and let $K_i(r) = \min\{\Delta_{i}(r), Q_{i}\}$ be an upper bound on the maximum demand that can be satisfied by facility~$i$ within {radius} $r$.

{First, we report the classical inequalities used to improve the linear relaxation of the CSCP-$r$:}
\begin{align}
\qquad x_{ij} \leq y_{i} & & i \in F, j \in C_i^r. \label{eq:CSCP3}
\end{align}

In addition, we use several techniques to improve {the CSCP-$r$ relaxation value} and reduce the number of nodes explored by the MILP branch-and-bound. 
{To illustrate these techniques, we consider a small instance with five vertices, each representing both a customer and a candidate location for a facility (i.e., $C=F$), with demands $q=[3, 5, 4, 1, 1]$, capacities $Q=[10, 15, 5, 15, 10]$, and positioned as illustrated in Figure \ref{fig:exampleIneq} at the end of this section.} The $d_{ij}$ values are assumed to be proportional to the Euclidean distances in the figure.\\

\noindent
\textbf{Facility domination inequalities.} Let $i_1$ and $i_2$ be two facilities with $C_{i_2}^r \subseteq C_{i_1}^r$ and $Q_{i_1} \geq K_{i_2}(r)$. Under these conditions, $i_1$ can be preferred to $i_2$ because {it can serve} a larger or equivalent set of customers. 
Consequently, an optimal solution that includes $i_2$ but not $i_1$ can always be transformed into an equivalent solution that includes $i_1$ but not $i_2$. We say that $i_1$ {\em dominates}~$i_2$. We thus obtain the following result:
\begin{proposition} \label{proposition:FDI}
The following inequalities are valid for the CSCP-$r$:
\begin{align}
y_{i_2} \leq y_{i_1} & \qquad i_1, i_2 \in F : C_{i_2}^r \subseteq C_{i_1}^r \text{ and } Q_{i_1} \geq K_{i_2}(r). \label{eq:ineq1}
\end{align}
\end{proposition}
Note that inequalities \eqref{eq:ineq1} do not forbid solutions in which both $i_1$ and $i_2$ are selected. In the example in Figure \ref{fig:exampleIneq},  facility 1 dominates facility 2, because it can serve the same customers of 2 and, in addition, customer 3.\\

\noindent
\textbf{Forcing service inequalities.} If $\Delta_{i}(r) \leq Q_i$  for a given facility $i$ and coverage radius~$r$, then the capacity constraint associated with $i$ becomes redundant. This means that whenever $i$ is open, all the customers $j \in C_{i}^r$ can be served by it. Hence, we have:
\begin{proposition} \label{proposition:FSI}
The following inequalities are valid for the CSCP-$r$:
\begin{align}
x_{ij} \geq y_{i} & \qquad i \in F : \Delta_{i}(r) \leq Q_i, j \in C_{i}^r. \label{eq:ineq2}
\end{align}
\end{proposition}
Note that the presence of two or more constraints \eqref{eq:ineq2} involving the same customer $j$ is allowed in our model because of the ``$\geq$'' in constraint \eqref{eq:CSCP2}. For the example in Figure \ref{fig:exampleIneq}, if facility 2 is open, then customers 1, 2, 4, and 5 are assigned to it.\\

\noindent
\textbf{Surplus demand inequalities.} Let $\Delta_i(r) > Q_i$ and $S \subset C_i^r$ be a subset of customers whose total demand does not exceed the capacity of $i$. We obtain:
\begin{proposition} \label{proposition:SDI}
The following inequalities are valid for the CSCP-$r$:
\begin{align}
x_{ik} \geq y_i - \bigg(\sum\nolimits_{j \in C_i^r \setminus S} x_{ij}\bigg) & \qquad i \in F, S \subset C_i^r : \sum\nolimits_{j \in S} q_j \leq Q_i, k \in S. \label{eq:ineq3}
\end{align}
\end{proposition}
Inequalities \eqref{eq:ineq3} generalize \eqref{eq:ineq2} because they force all customers $k \in S$ to be served by $i$ when~$i$ is open, and all customers $j \in C_i^r \setminus S$ are served by a facility other than $i$. In the example of Figure \ref{fig:exampleIneq}, this inequality occurs for $i=1$ and $S=\{1,3,5\}$, imposing customers 1, 3, and 5 to be assigned to facility 1 if $x_{12}=x_{14}=0$ and {$y_1 = 1$}.\\

\noindent
\textbf{Symmetry breaking inequalities.} Consider the case in which two customers $j_1$ and $j_2$, {with identical demands}, are located in the common coverage area of two facilities $i_1$ and $i_2$, and  suppose  $j_1 < j_2$ and $i_1 < i_2$. Then, one can forbid the assignment of $j_1$ to $i_2$ and of $j_2$ to $i_1$, since this would be equivalent to assigning $j_1$ to $i_1$ and $j_2$ to $i_2$.
Consequently:
\begin{proposition} \label{proposition:BSI}
The following inequalities are valid for the CSCP-$r$ 
\begin{align}
x_{i_1,j_2} + x_{i_2,j_1} \leq 1 & \qquad  i_1, i_2 \in F,  j_1, j_2 \in C_{i_1}^r \cap C_{i_2}^r: i_1 < i_2, j_1 < j_2, q_{j_1} = q_{j_2}. \label{eq:ineq4}
\end{align}
\end{proposition}
Inequalities \eqref{eq:ineq4} do not forbid the assignments of both customers to the same facility (or to neither $i_1$ nor $i_2$).
In the example in Figure \ref{fig:exampleIneq}, the assignment of 4 to 2 and 5 to 1 is forbidden, but the equivalent assignment of 4 to 1 and 5 to 2 is still possible. 
Note that these inequalities may be incompatible with \eqref{eq:ineq2} and \eqref{eq:ineq3}, which might force an assignment instead of forbidding it. To avoid this issue, we simply remove from constraints \eqref{eq:ineq2} and \eqref{eq:ineq3} all inequalities containing a variable that also appears in an inequality of type \eqref{eq:ineq4}.\\

\noindent
\textbf{Capacity inequalities.} Let $S \subseteq C$ be a subset of customers and $F_{S} = \bigcup_{j \in S} F_j$ be the subset of facilities covering at least one customer in $S$. 
If all facilities in $F_{S}$ had the same capacity, say, $Q$, then a lower bound on the minimum number of facilities required to serve $S$ could be imposed by adding
\begin{align*}
\sum\nolimits_{k \in F_{S}} y_{k} \geq \bigg\lceil {\sum\nolimits_{j \in S} q_j}/{Q} \bigg\rceil & \qquad S \subseteq C. 
\end{align*}
Following the work on the heterogeneous vehicle routing problem by \citet{Yaman2006}, {these constraints} can be extended to handle heterogeneous capacities using the next result: 
\begin{proposition} \label{proposition:CI}
The following inequalities are valid for the CSCP-$r$:
\begin{align}
\sum\nolimits_{k \in F_{S}} \lceil {Q_{k}}/{\gamma} \rceil y_{k} \geq \bigg\lceil {\sum\nolimits_{j \in S} q_j}/{\gamma} \bigg\rceil & \qquad S \subseteq C, \gamma \in \mathbb{N}. \label{eq:ineq5}
\end{align}
\end{proposition}
In our implementation, we consider a facility $i$, select all subsets $S$ containing 2 or 3 customers and lying entirely in the coverage area of $i$, then compute $F_S$ and impose $\gamma = Q_i$ in \eqref{eq:ineq5}. The process is repeated for all $i \in F$. 
In Figure \ref{fig:exampleIneq}, the example for inequality~\eqref{eq:ineq5} is obtained by setting $S=\{1,2,3\}$, $F_{S}=\{1,2,3,4,5\}$, and $\gamma = Q_1$.\\

\noindent
{\textbf{Subset sum inequalities.}} From \citet{Boschetti2002}, the capacity of a facility $i$ can potentially be decreased, while preserving optimality, by computing the maximal capacity usage through the solution of the following \emph{subset sum problem} (SSP): 
\begin{equation*}
Q'_i=\max \bigg\{ \sum\nolimits_{j \in C_i^r} q_j x_{ij} : \sum\nolimits_{j \in C_i^r} q_j x_{ij} \leq Q_i,  x_{ij} \in \{0,1\} \text{ for } j \in C_i^r \bigg\}. 
\end{equation*}
Since $Q'_i \leq Q_i$, the capacity constraints \eqref{CpCP:const5} can be improved to  
\begin{align}
\sum\nolimits_{j \in C_{i}^r} q_j x_{ij} \leq Q'_i y_i & \qquad i \in F. \label{BMIneq1}
\end{align}
The same idea can be applied to increase the customers' demands. The demand of a customer~$k$, when assigned to a facility $i$, can be increased by evaluating the maximal usage of the facility capacity through the solution of the SSP
\begin{equation*}
\resizebox{.92\hsize}{!}{$
\beta^{1}_{ik}=\max \Bigg\{ \sum\nolimits_{j \in C_i^r \setminus \{k\}} q_j x_{ij} : \sum\nolimits_{j \in C_i^r \setminus \{k\}} q_{j} x_{ij} \leq Q_i - q_k, \text{ } x_{ij} \in \{0,1\} \text{ for } j \in C_i^r \setminus \{k\} \Bigg\}. 
$}
\end{equation*}
If $\beta^1_{ik}$ is lower than $Q_i - q_k$, then the unused capacity in  $i$ amounts to at least $Q_i - q_k - \beta^1_{ik}$ units, so $q_k$ can be increased to $q_k + Q_i - q_k - \beta^1_{ik} = Q_i - \beta^1_{ik}$, leading to 
\begin{proposition} \label{proposition:SSI}
The following inequalities are valid for the CSCP-$r$:
\begin{align}
\sum\nolimits_{j \in C_{i}^r \setminus \{k\}} q_j x_{ij} + (Q_i - \beta^1_{ik})x_{ik} \leq Q_i y_i & \qquad i \in F, k \in C_i^r. \label{BMIneq2}
\end{align}
\end{proposition}
These cuts have been used intensively in the cutting and packing literature. However, it is easy to find cases in which they do not provide any improvement with respect to the original capacity constraints \eqref{CpCP:const5}. This is highlighted by the example in Figure \ref{fig:exampleIneq}, in which both \eqref{BMIneq1} for facility 1, and \eqref{BMIneq2} for facility 1 and customer 2 are equivalent to  \eqref{CpCP:const5}. This fact holds for all possible cuts in the example. {The following inequalities provide, instead, a stronger improvement over \eqref{CpCP:const5}.}\\

\noindent
{\textbf{Disjunctive subset sum inequalities.}}
As discussed earlier in this section, $\beta^{1}_{ik}$ represents the maximum possible use of capacity $Q_i$ when $k$ is served by $i$. Now let $\beta^{0}_{ik}$ be the maximum $i_k$ use when $k$ is not served by $i$, namely:
\begin{equation*}
\beta^{0}_{ik}=\max \Bigg\{ \sum\nolimits_{j \in C_i^r \setminus \{k\}} q_j x_{ij} : \sum\nolimits_{j \in C_i^r \setminus \{k\}} q_{j} x_{ij} \leq Q_i, \text{ } x_{ij} \in \{0,1\} \text{ for } j \in C_i^r \setminus \{k\} \Bigg\}. 
\end{equation*}
{An extension of the subset sum inequalities \eqref{BMIneq2} can be performed by} embedding these $\beta^{0}_{ik}$ values in the trivial disjunctive cuts by \cite{Balas1973} (in our case being $x_{ij} \leq 0 \lor x_{ij} \geq 1$ for $i \in F, j \in C_i^r$),
as follows:
\begin{proposition} \label{proposition:DBSSI}
The following inequalities are valid for the CSCP-$r$:
\begin{align}
\qquad \sum\nolimits_{j \in C_i^r \setminus \{k\}} q_{j} x_{ij} \leq \beta^{0}_{ik}y_{i} - (\beta^{0}_{ik}-\beta^{1}_{ik})x_{ik} & \qquad i \in F, k \in C_i^r. \label{ourSSineq}
\end{align}
\end{proposition}
\begin{proof}
  We wish to prove that inequalities \eqref{ourSSineq} correspond to the disjunction
  \begin{align*}
  \qquad \sum\nolimits_{j \in C_i^r \setminus \{k\}} q_{j} x_{ij} \leq \beta^{0}_{ik}y_{i} \qquad  \lor  \qquad  \sum\nolimits_{j \in C_i^r \setminus \{k\}} q_{j} x_{ij} \leq \beta^{1}_{ik}y_{i} & \qquad i \in F, j \in C_i^r.
  \end{align*}
  We consider the two cases in which $x_{ik}$ takes value 0 (left) or 1 (right). When $x_{ik}=0$, constraints \eqref{ourSSineq} directly reduce to
  \begin{align*}
  \qquad \sum\nolimits_{j \in C_i^r \setminus \{k\}} q_{j} x_{ij} \leq \beta^{0}_{ik} y_{i},
  \end{align*}
  which is valid because all $x_{ij}$ take value 0 when $y_i=0$ and because of the maximality of $\beta^{0}_{ik}$.
  If $x_{ik}=1$ instead, then $y_i$ must take value 1 too, and thus constraints \eqref{ourSSineq} can be~rewritten~as
  \begin{align*}
  \qquad \sum\nolimits_{j \in C_i^r \setminus \{k\}} q_{j} x_{ij} \leq \beta^{0}_{ik} y_i - (\beta^{0}_{ik} - \beta^{1}_{ik}) \leq \beta^{0}_{ik} y_i - (\beta^{0}_{ik} - \beta^{1}_{ik}) y_i = \beta^{1}_{ik} y_i,
  \end{align*}
  which is valid because of the maximality of $\beta^{1}_{ik}$.
\end{proof}

The idea behind inequalities \eqref{ourSSineq} is reminiscent of \emph{up-lifting} and \emph{down-lifting} in cutting plane theory. For instance, in a binary integer problem, a given inequality can be strengthened by considering the linear relaxation solution at a certain {node of a branch-and-bound tree and taking into account the previously performed branches to 0 or 1} (see, e.g., \citealt{Kaparis2008}, and \citealt{Vasilyev2016}).

In Figure \ref{fig:exampleIneq}, the inequality \eqref{ourSSineq} applied to facility 1 and customer 2 leads to a stronger capacity constraint. 
Inequalities \eqref{ourSSineq}  are indeed more effective than the well-known inequalities \eqref{BMIneq2} but also more specialized: the constraints \eqref{BMIneq2} can be readily extended to integer $x$ variables, whereas the constraints \eqref{ourSSineq} hold only for binary $x$ variables due to the nature of the disjunction.

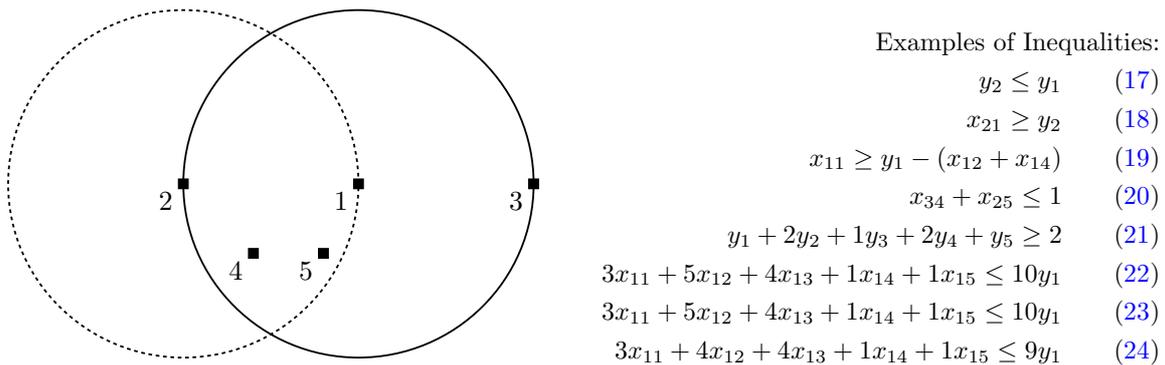
\begin{figure}[htbp]
\centering
\scalebox{0.97}
{
\begin{tabular}{cc}
\begin{tikzpicture}[scale=1, every node/.style={scale=3}]
    \coordinate (O) at (0,0);

    \draw[line width=1.5pt,dashed] (5,10) circle (5cm);
    \draw[line width=1.5pt] (10,10) circle (5cm);

    \node[draw, rectangle, scale=0.4, fill=black ] at (5,10) {}; \node[scale=0.75] at (4.5,9.5) {2};
    \node[draw, rectangle, scale=0.4, fill=black ] at (10,10) {}; \node[scale=0.75] at (9.5,9.5) {1};
    \node[draw, rectangle, scale=0.4, fill=black ] at (15,10) {}; \node[scale=0.75] at (14.5,9.5) {3};
    \node[draw, rectangle, scale=0.4, fill=black ] at (7,8) {}; \node[scale=0.75] at (6.5,7.5) {4};
    \node[draw, rectangle, scale=0.4, fill=black ] at (9,8) {}; \node[scale=0.75] at (8.5,7.5) {5};
\end{tikzpicture} \hspace{0.2cm}
& 
$\begin{aligned}[b]
\text{Examples of Inequalities:} \\
y_2 \leq y_1 \qquad (\ref{eq:ineq1})\\ 
x_{21} \geq y_2 \qquad (\ref{eq:ineq2})\\ 
x_{11} \geq y_1 - (x_{12}+x_{14}) \qquad (\ref{eq:ineq3})\\ 
x_{34} + x_{25} \leq 1 \qquad (\ref{eq:ineq4})\\ 
y_1 + 2 y_2 + 1 y_3 + 2 y_4 + y_5 \geq 2 \qquad (\ref{eq:ineq5})\\ 
3x_{11} + 5x_{12} + 4x_{13} + 1x_{14} + 1x_{15} \leq 10y_1 \qquad (\ref{BMIneq1})\\ 
3x_{11} + 5x_{12} + 4x_{13} + 1x_{14} + 1x_{15} \leq 10y_1 \qquad (\ref{BMIneq2})\\ 
3x_{11} + 4x_{12} + 4x_{13} + 1x_{14} + 1x_{15} \leq 9y_1  \qquad (\ref{ourSSineq})
\end{aligned}$ 
\\
\end{tabular}
}
\caption{Examples of valid inequalities for a small instance. Input data: $F = C = \{1, 2, 3, 4, 5\}$; $Q=$[10, 15, 5, 15, 10]; and $q=$[3, 5, 4, 1, 1].}
\label{fig:exampleIneq}
\end{figure}

\section{Searching for the Optimal Radius} \label{sec:method}

As presented in Section \ref{subsec:decomposition}, an optimal CPCP solution can be found by iteratively solving a series of CSCP-$r$ subproblems, each associated with a different distance. Let $(z_1, \dots, z_D)$ be the distinct values from the distance matrix, in increasing order. Then, $z_{k}$, for $k \in \{1, \dots,D\}$, is the optimal distance value if and only if there exists a feasible solution for the CSCP-$r$ subproblem with $r=z_{k}$, and either $k=1$ or there is no feasible solution for $r=z_{k-1}$. {Apart from the trivial case where $z_1$ is optimal}, at least two subproblems, one feasible and one infeasible, have to be solved to prove optimality, but more subproblems are usually solved because the optimal value is unknown. \\

\noindent
\textbf{Binary and sequential search.}
The most straightforward approach to find the optimal distance is to adopt a \emph{binary search} over the vector $(z_{i^\textsc{low}},\dots,z_{i^\textsc{up}})$, where $z_{i^\textsc{low}}$ is a strict lower bound on the distance (known to be infeasible) and $z_{i^\textsc{up}}$ is an upper bound (known to be feasible). Iteratively, the subproblem obtained for $z_{i^\textsc{mid}}$ with $i^\textsc{mid} =  \lfloor \frac{i^\textsc{low}+i^\textsc{up}}{2} \rfloor$ is solved. If the subproblem is feasible, then the next iteration is performed over $(z_{i^\textsc{low}},\dots,z_{i^\textsc{mid}})$, otherwise the search is performed over $(z_{i^\textsc{mid}},\dots,z_{i^\textsc{up}})$. The process terminates when $i^\textsc{up} = i^\textsc{low}+1$. This search strategy is guaranteed to take $O(\log D) = O(\log (n m))$ iterations. This is the best possible method in terms of number of search iterations. However, the two drawbacks of this approach are that it tends to \emph{overshoot} the value of the optimal solution in the early steps of the search, and that it solves on average 50\% of feasible and infeasible subproblems. In our context, the size of the mathematical formulation solved at each iteration grows quickly with $z$, because more variables and constraints have to be considered. Moreover, in our experiments it appeared to be more computationally expensive to solve a feasible problem (i.e., find a feasible CPCP solution) than an infeasible one (i.e., to prove infeasibility).

In light of the previous observations, the second most natural strategy involves a \emph{sequential search}, solving the subproblems in increasing order of $z \in (z_{i^\textsc{low}},\dots,z_{i^\textsc{up}})$ and stopping as soon as a feasible solution is found. This strategy circumvents the two aforementioned issues, because it avoids large problems with $z$ values that are greater than the optimum. It also requires the solution of only one feasible subproblem. However, its drawback is the high number of iterations, which rises to $O(n m)$ and renders this approach slower for problems with many distinct distance values and a weak lower bound.\\

\noindent
\textbf{Layered search.}
We {\small thus propose an} alternative search methodology, called \textsc{L-Layered Search} and described in Algorithm~\ref{alternativeSearch}, which combines the benefits of both previous approaches. This method can be viewed as a recursive sequential search, starting with larger increments \mbox{(when $L > 1$)}, and finishing with smaller ones (a classical sequential search when $L=1$). At each step, the algorithm solves a sequence of subproblems with increasing distance bounds \mbox{(lines 5--8)}, and then performs a recursive call with parameter $L-1$ on a smaller interval \mbox{(lines 9--13)}. 
The algorithm is initially called with the function \textsc{Layered Search}$(L,i^\textsc{low},i^\textsc{up})$, and it terminates as soon as $i^\textsc{up} = i^\textsc{low}+1$ (line 1). 
The increment size of $\delta =  \lceil (i^\textsc{up}-i^\textsc{low}-1)^{(L-1)/L} \rceil$ is chosen at each recursion (line 2) to attain good overall complexity and guarantee a small number of calls to feasible subproblems.

\begin{figure}[ht]
  \centering
  \begin{minipage}{1.0\linewidth}
\begin{algorithm}[H]
\linespread{1.25}\selectfont

 \textbf{if} $i^\textsc{up} = i^\textsc{low}+1$ \textbf{then return} $z_{i^\textsc{up}}$ \textbf{end if}

 $\delta \leftarrow \lceil (i^\textsc{up}-i^\textsc{low}-1)^{(L-1)/L} \rceil$ \label{delta}

$i \gets i^\textsc{low} $

$isFeasible \leftarrow \textsc{False}$

\While {$isFeasible= \textsc{False}$ \textbf{and} $i + \delta < i^\textsc{up}$} 
{

	$i \leftarrow i+\delta$

	$isFeasible \leftarrow$ \textsc{SolveCSCP}-$r$($z_{i}$)    
}

\eIf{$isFeasible= \textsc{True}$}
{
\textsc{Layered Search}$(L-1,i-\delta,i)$
}
{
\textsc{Layered Search}$(L-1,i,i^\textsc{up})$
}

\caption{\textsc{Layered Search}$(L,i^\textsc{low},i^\textsc{up})$} \label{alternativeSearch}
\end{algorithm}
 \end{minipage}
\end{figure}

\begin{proposition} \label{proposition-layered1}
The number of feasible subproblems solved by \textsc{Layered Search}$(L,i^\textsc{low},$ $i^\textsc{up})$ is in $O(L)$.
\end{proposition}
\begin{proof}
{At each recursive call, at most one feasible subproblem is solved in the loop of lines 5--8, and~$L$ is decremented by one unit. Moreover, when $L=1$ the behavior of the algorithm corresponds to a sequential search (because $\delta = 1$ at step 2), and this necessarily leads to the termination criterion for $L=0$, if that was not already attained.} 
\end{proof}

\begin{proposition} \label{property-layered2}
The number of subproblems solved by \textsc{Layered Search}$(L,i^\textsc{low},i^\textsc{up})$ is in \mbox{$O(L (i^\textsc{up} - i^\textsc{low} - 1)^{1/L})$}.
\end{proposition}
\begin{proof}
Let $N =  i^\textsc{up} - i^\textsc{low} - 1$ be the number of subproblems with unknown status within the range $\{i^\textsc{low},\dots,i^\textsc{up}\}$, and let $T(L,N)$ be the complexity of \textsc{Layered Search}$(L, i^\textsc{low}, i^\textsc{up})$ in terms of number of subproblem resolutions. The algorithm solves up to
$
\frac{N}{ \lceil N^{(L-1)/L} \rceil} \leq N^{1/L}
$
subproblems, and then performs a recursive call on a smaller range, in which the number of subproblems with unknown feasibility status is reduced to $\lceil N^{(L-1)/L} \rceil - 1$. Therefore, the following holds: 
\begin{equation}
T(L,N) \leq
\begin{cases}
N & \text{ if } L = 1 \\
N^{1/L} + T(L-1,\lceil N^{(L-1)/L} \rceil - 1) & \text{ if } L > 1, \nonumber \\ 
\end{cases}
\end{equation}
and $T(L,N) \leq L N^{1/L}$ follows by direct induction. 
\end{proof}

The complexity of the three search strategies, in terms of total number of subproblem calls and number of calls to feasible subproblems, is summarized in Table \ref{tab:strat}. One can observe that the L-layered search constitutes an interesting alternative between sequential and binary search, behaving as a sequential search when $L=1$, and coming closer to a binary search (low number of subproblems overall, but no control on the number of feasible subproblems) as $L$ grows larger. Typical values for $L$ are in the range $\{2,\dots,5\}$.

\clearpage

\begin{table}[htbp]
\renewcommand{\arraystretch}{1.5}
\centering
\caption{{Complexity of search strategies (${N}$ = number of distinct values in the range to be searched)}}
\scalebox{.95}
{
\begin{tabular}{cccccc}
\toprule
& {Sequential} & {2-Layers} & {3-Layers} & {L-Layers} & {Binary} \\
\cmidrule(r){1-1}\cmidrule(r){2-6}
{\# overall subproblems}  & $O(N)$ & $O(\sqrt{N})$ & $O(N^{1/3})$ & $O(LN^{1/L})$ &  $O(\log N)$ \\
{\# feasible subproblems} & 1 & 2 & 3 & L & $O(\log N)$ \\
\bottomrule
\end{tabular}
}
\label{tab:strat}
\end{table} 

\noindent
\textbf{Other speedup techniques.}
{A complete solution of the CSCP-$r$ subproblem is often unnecessary. To reduce the computational effort, we use three simple techniques: 
(i) we estimate the maximum demand that can be covered by each facility by looking at the customers in its coverage area, and then, if the sum of the $p$ largest such values is lower than the total demand, we avoid solving the CSCP-$r$ (since the subproblem is clearly infeasible); (ii) we terminate the solution of each CSCP-$r$ whenever a valid lower bound greater than $p$ is achieved (infeasible subproblem); (iii) we terminate the solution of the CSCP-$r$ whenever a feasible solution opening at most $p$ facilities is found (feasible subproblem).}

In addition, we obtain initial upper bounds on the optimal radius via the heuristic of \citet{QuevedoOrozco2015}, which is the current state-of-the-art metaheuristic for the CPCP. However, as reported by the authors, this algorithm does not perform well on instances with a small $n/p$ ratio. Thus, we also perform a simple {\em iterated local search} (ILS), possibly obtaining a better initial upper bound.

Our simple ILS builds an initial solution by opening one facility at a time and assigning customers as follows.
Let $\bar{C}$ be the set of unassigned customers.
The facility $i$ that realizes the minimum of $\max_{j \in \bar{C}}d_{ij}/Q_i$ is opened, 
and the closest customers in $\bar{C}$ are assigned to~$i$, iteratively, 
until the capacity constraints prevent further assignments. This process is repeated until a maximum of $p$ facilities have been opened. If some customers remain unassigned, then they are assigned to their closest facility, possibly giving an infeasible initial solution. 

Subsequently, the method attempts to improve this solution with local search and perturbation steps. Let $\bar Q_i$ be the excess of capacity in facility $i$ (either zero for a feasible solution, or strictly positive for an infeasible solution), and $\bar{d}_{ij} = d_{ij} + M \bar Q_i$ be the cost function associated with the assignment of customer~$j$ to facility $i$, where $M$ is a large value. The ILS attempts to change the customer and facility associated with the largest $\bar{d}_{ij}$. Three types of moves are considered: \textsc{Cust-Swap}, \textsc{Relocate}, and \textsc{Fac-Swap}, invoked in this order. \textsc{Cust-Swap} exchanges two customers assigned to different facilities, \textsc{Relocate} attempts to move a customer from one facility to another, and \textsc{Fac-Swap} exchanges an open facility with a closed one.
Finally, a strong perturbation operator is invoked once no improving move can be found. 
{This perturbation randomly exchanges the open facilities with the closed ones (one exchange for each open facility).} 
The algorithm terminates after 300 executions of the local search and perturbation mechanism.

This simplistic metaheuristic was able to produce a feasible solution for all the instances presented in Section \ref{sec:instances}, and it led to a better upper bound in 73 instances. The detailed results are provided on our webpage: \url{http://www.or.unimore.it/site/home/online-resources.html}.

\FloatBarrier
\section{Computational Experiments}
\label{sec:results}

{We conducted extensive computational experiments to evaluate the performance of the proposed techniques. The algorithms were coded in C++ and executed on a single thread of an Intel Core i5-2410M 2.3GHz with 4GB of RAM, running under Linux Mint 17.2 64-bit. GUROBI 6.51 was adopted to solve the MILP models, using its default parameters. A time limit of 600 seconds was allowed for each run.}

\subsection{Benchmark instances} \label{sec:instances}

{We considered five sets of instances: four are from previous literature, and the fifth is a set of large instances that we created to better evaluate our methods. Following the literature, in all the test sets each vertex is both a customer and a candidate location for a facility (i.e., $C=F$).
\begin{itemize}
 \item Set 1 (S1) -- This set contains $160$ instances proposed for the CPMP by \citet{Ceselli2005} and derived from $40$ different graphs containing either $50$, $100$, $150$, or $200$ vertices. {This set includes and extends 20 instances from the OR-Library used in \citet{Scaparraetal2004}.} 
 The demands and the coordinates of the vertices were randomly generated, and the distances were obtained by computing Euclidean values rounded down to the nearest integer. From each graph, four instances were produced by setting $p$ equal to $\lfloor n/10 \rfloor$, $\lfloor n/4 \rfloor$, $\lfloor n/3 \rfloor$, or $\lfloor n/2.5 \rfloor$. The facility capacities are homogeneous and set to $\lceil 12n/p \rceil$.
 \item Set 2 (S2) -- This set contains eight instances by \citet{Scaparraetal2004}, derived from two different graphs containing either $100$ or $150$ vertices and with non-Euclidean integer distances. From each graph, four instances with homogeneous capacities and four with heterogeneous capacities were created by selecting $p \in \{5,15\}$. 
 \item Set 3 (S3) -- Proposed by \citet{LorenaSenne2004} for the CPMP, this set contains six instances with $n$ varying from $100$ to $402$ and $p$ from $10$ to $40$, with 
 homogeneous capacities equal to $\Big\lceil {\sum_j q_j}/{(\tau p)} \Big\rceil$ and $\tau \in \{0.8,0.9\}$. The distances are Euclidean. \citet{AlbaredaSambola2010} considered floating-point values, while \citet{QuevedoOrozco2015} rounded down to the nearest integer.
 \item Set 4 (S4) -- This set contains five instances by \citet{LorenaSenne2004}, obtained by modifying the {\em Pcb3038} instance of the TSPLIB, varying $p$ in $\{600, 700, 800, 900,$ $1000\}$. The facility capacities are homogeneous and set to $\Big\lceil {\sum_j q_j}/{(\tau p)} \Big\rceil$, with $\tau \in \{0.8,0.9\}$.
 \item New set (KIV) -- We generated 280 new instances that are similar to S1, but have either $n \in \{300, 500,1000, 2000, 3000\}$ and integer distances, or $n \in \{50, 100, 150, 200,$ $300, 500, 1000, 2000, 3000\}$ and floating-point distances. {As in S1, the demand values were randomly generated within the interval $[1,20]$.} To obtain a customer density similar to that of \citet{Scaparraetal2004}, we randomly generated the vertex coordinates in $[1, \sqrt{100n}]$. We created 20 instances for each value of $n$, dividing them into four groups of five instances each, with $p$ equal to $\left\lfloor\frac{n}{10} \right\rfloor$, $\left\lfloor\frac{n}{7}\right\rfloor$, $\left\lfloor\frac{n}{4}\right\rfloor$, or $\left\lfloor\frac{n}{3}\right\rfloor$. These instances are available at \url{http://www.or.unimore.it/site/home/online-resources.html}.
\end{itemize}
}

\subsection{Evaluation of the solution techniques for the CSCP subproblem} \label{sec:eval-formulations}

Our first experiment focused on the formulations and valid inequalities used to solve the CSCP subproblem. 
Recall that the subproblem seeks to minimize the number of facilities needed to cover the demands subject to a radius limit $r$, and that finding any lower bound greater than $p$, or a feasible solution using up to $p$ facilities, immediately terminates the subproblem execution.
For this experiment, we used S1 and considered $r= \text{BLB}-1$, where BLB is the best known lower bound from the literature, collected from \citet{QuevedoOrozco2015}. {Based on our computational experiments, we know that this BLB matches the optimal radius in 150 of the 160 instances.
This choice of $r$ is particularly relevant because the solution of the associated subproblem usually constitutes the last and hardest infeasible iteration. A formulation that produces good {\em linear programming} (LP) relaxation values in reasonable times is clearly preferable.}\\

\noindent \textbf{Choice of the CSCP formulation.}
{Table \ref{tab:evalCSCP} presents the results of our first experiment, which compares the descriptive and arc-flow formulations (CSCP-$r$ and CSCP-AF-$r$, respectively), with and without additional inequalities. Each row refers to a group of ten instances with the same  $n$ and $p$ values. The left part of the table presents the results of the LP relaxations of the ``plain'' formulations, obtained by disregarding the inequalities of Section \ref{sec:improvements}, namely, \eqref{eq:CSCPinit}--\eqref{eq:CSCPend} for CSCP-$r$ and \eqref{eq:AFinit}--\eqref{eq:CSCP2}, \eqref{eq:CSCP6}, \eqref{eq:CSCPend} for CSCP-AF-$r$. The right part shows the results achieved by adding the inequalities \eqref{eq:CSCP3}--\eqref{eq:ineq5} and \eqref{ourSSineq} to the plain formulations, thus obtaining the ``full'' formulations. Column ``LB'' gives the average LP relaxation lower bound, ``$> p$'' reports the number of instances for which $LB>p$, ``T(s)'' gives the average CPU seconds required to solve the LP models, and ``gap(\%)'' is the percentage gap between the two LP relaxations (computed as $100 \times ({LB}_B - LB_A)/LB_A$, where $LB_A$ and $LB_B$ are the LP bounds of CSCP-AF-$r$ and CSCP-$r$, respectively).}

\begin{table}[htbp]
\centering
\setlength{\tabcolsep}{4.4pt}
\renewcommand{\arraystretch}{1.3}
\caption{Comparison of the linear relaxations of plain and full CSCP-AF-$r$ and CSCP-$r$ formulations.}
{
\begin{adjustbox}{max width=\textwidth}
\begin{tabular}{rrrrrrrrrrrrrrrrrrrrr}
\toprule
\multicolumn{ 3}{c}{} &  & \multicolumn{ 8}{c}{plain formulations (LP)} & \multicolumn{1}{l}{} & \multicolumn{ 8}{c}{full formulations (LP)} \\ \cline{5-12} \cline{14-21}
\multicolumn{ 3}{c}{\uppp{instance}} &  & \multicolumn{ 3}{c}{CSCP-AF-$r$} & \multicolumn{1}{l}{} & \multicolumn{ 3}{c}{CSCP-$r$} && \multicolumn{1}{l}{} & \multicolumn{ 3}{c}{CSCP-AF-$r$} & \multicolumn{1}{l}{} & \multicolumn{3}{c}{CSCP-$r$} \\ \cline{1-3} \cline{5-7} \cline{9-11} \cline{14-16} \cline{18-20}
$n$ & $p$ & \# &  & LB & $> p$ & T(s) &  & LB & $> p$ & T(s) & \uppp{gap(\%)} &  & LB & $> p$ & T(s) &  & LB & $> p$ & T(s) & \uppp{gap(\%)} \\ 
 \cline{1-3} \cline{5-12} \cline{14-21}
 50 & 5 & 10 &  & 4.69 & 2 & 2.65 &  & 4.46 & 0 & 0.07 & -4.90 &  & 5.28 & 8 & 3.77 &  & 5.28 & 8 & 0.44 & 0.00 \\ 
100 & 10 & 10 &  & 9.22 & 0 & 5.51 &  & 8.82 & 0 & 0.14 & -4.29 &  & 10.10 & 5 & 12.61 &  & 10.10 & 5 & 1.50 & 0.00 \\ 
150 & 15 & 10 &  & 13.72 & 0 & 13.44 &  & 12.92 & 0 & 0.21 & -5.83 &  & 15.67 & 7 & 25.94 &  & 15.67 & 7 & 3.18 & 0.00 \\ 
200 & 20 & 10 &  & 18.40 & 0 & 28.58 &  & 17.58 & 0 & 0.33 & -4.49 &  & 20.41 & 7 & 59.14 &  & 20.41 & 7 & 7.75 & -0.01 \\ 
 \cline{1-3} \cline{5-12} \cline{14-21}
50 & 12 & 10 &  & 11.75 & 2 & 0.17 &  & 10.70 & 0 & 0.02 & -9.00 &  & 12.04 & 6 & 0.51 &  & 12.00 & 6 & 0.16 & -0.32 \\ 
100 & 25 & 10 &  & 24.18 & 2 & 0.34 &  & 22.05 & 0 & 0.04 & -8.78 &  & 25.13 & 6 & 1.15 &  & 25.05 & 5 & 0.29 & -0.30 \\ 
150 & 37 & 10 &  & 36.21 & 2 & 0.52 &  & 31.64 & 0 & 0.07 & -12.61 &  & 38.13 & 10 & 1.91 &  & 37.95 & 10 & 0.43 & -0.46 \\ 
200 & 50 & 10 &  & 49.34 & 2 & 0.90 &  & 43.95 & 0 & 0.11 & -10.92 &  & 51.25 & 9 & 3.36 &  & 51.02 & 8 & 0.82 & -0.45 \\ 
 \cline{1-3} \cline{5-12} \cline{14-21}
50 & 16 & 10 &  & 15.78 & 2 & 0.08 &  & 14.07 & 0 & 0.01 & -10.82 &  & 16.12 & 6 & 0.30 &  & 16.03 & 5 & 0.15 & -0.54 \\ 
100 & 33 & 10 &  & 32.65 & 3 & 0.19 &  & 28.61 & 0 & 0.03 & -12.36 &  & 33.30 & 4 & 0.66 &  & 32.88 & 4 & 0.24 & -1.25 \\ 
150 & 50 & 10 &  & 49.75 & 3 & 0.31 &  & 43.07 & 0 & 0.06 & -13.42 &  & 51.08 & 9 & 0.97 &  & 50.50 & 6 & 0.36 & -1.13 \\ 
200 & 66 & 10 &  & 67.16 & 7 & 0.54 &  & 57.01 & 0 & 0.10 & -15.11 &  & 68.99 & 9 & 1.55 &  & 68.16 & 8 & 0.61 & -1.20 \\ 
 \cline{1-3} \cline{5-12} \cline{14-21}
50 & 20 & 10 &  & 19.86 & 4 & 0.07 &  & 17.83 & 0 & 0.01 & -10.24 &  & 20.12 & 6 & 0.36 &  & 19.58 & 1 & 0.27 & -2.69 \\ 
100 & 40 & 10 &  & 39.83 & 2 & 0.15 &  & 35.29 & 0 & 0.03 & -11.40 &  & 40.19 & 4 & 0.63 &  & 39.07 & 2 & 0.33 & -2.77 \\ 
150 & 60 & 10 &  & 59.77 & 3 & 0.27 &  & 51.68 & 0 & 0.06 & -13.53 &  & 60.71 & 5 & 0.96 &  & 58.72 & 3 & 0.56 & -3.28 \\ 
200 & 80 & 10 &  & 81.58 & 8 & 0.44 &  & 70.32 & 0 & 0.09 & -13.81 &  & 82.92 & 8 & 1.86 &  & 80.30 & 7 & 0.93 & -3.16 \\ 
 \cline{1-3} \cline{5-12} \cline{14-21}
\multicolumn{ 3}{c}{avg/sum} &  & 33.37 & 42 & 3.39 &  & 29.38 & 0 & 0.09 & -10.10 &  & 34.46 & 109 & 7.23 &  & 33.92 & 92 & 1.13 & -1.10 \\ 
\bottomrule
\end{tabular}
\end{adjustbox}
}
\label{tab:evalCSCP}
\end{table}

{From Table \ref{tab:evalCSCP}, we first observe that the plain CSCP-AF-$r$ provides significantly better LP bounds than the plain CSCP-$r$. For 42 instances, infeasibility is proven by simply solving the plain CSCP-AF-$r$ LP. The drawback of CSCP-AF-$r$ is its higher CPU time consumption, especially on instances with large~$n/p$ ratio. The full CSCP-AF-$r$ also provides better lower bounds than the full CSCP-$r$. However, the CSCP-$r$ benefits more from the additional inequalities: the gap  between the two bounds reduces from $-10.10\%$ to $-1.10\%$ when the additional inequalities are included. The full CSCP-$r$ LP bound is sufficient to prove CPCP infeasibility for 92 instances, compared to 109 instances for the full CSCP-AF-$r$. However, solving the full CSCP-AF-$r$ LP requires six times the CPU effort required by the full CSCP-$r$ LP. Given the similar performance of the two full formulations, but their significant difference in terms of time, we decided to invoke the CSCP-$r$ at each iteration of our search algorithms.}\\

\noindent
\textbf{Impact of the valid inequalities.}
{Our second experiment, reported in Table~\ref{tab:evalLBs}, analyzes in detail the impact of the inequalities from Section \ref{sec:improvements}. In the left part of the table, column ``LB'' provides the average LP relaxation value of the plain CSCP-$r$. Each successive column under the label ``gap(\%)'' reports the gap between LB and the LP relaxation value of the plain CSCP-$r$ after the inclusion of the indicated inequality. Similarly, each column in the right part of the table shows the gap(\%) from the average LP relaxation value of the full CSCP-$r$ when removing the indicated inequality.}
{From Table \ref{tab:evalLBs}, we observe that the addition of inequalities \eqref{eq:ineq3} to the plain \mbox{CSCP-$r$}  does not improve the LP bound, but their removal from the full formulation decreases the LP bound, even if only slightly. All the other inequalities have a visible impact on the LB when added to the plain formulation or removed from the full formulation. The well-known inequalities \eqref{eq:CSCP3} are essential for the quality of the LB, as are the new disjunctive subset sum inequalities \eqref{ourSSineq}. Inequalities \eqref{eq:ineq3}, \eqref{eq:ineq5}, and \eqref{ourSSineq} have a larger impact on instances with small $n/p$ values, while inequalities \eqref{eq:CSCP3} contribute more on instances with large $n/p$ values.}\\
\begin{table}[htbp]
\centering
\setlength{\tabcolsep}{2.5pt}
\renewcommand{\arraystretch}{1.3}
\caption{Impact of valid inequalities on linear relaxation of CSCP-$r$ formulation.}
{
\begin{adjustbox}{max width=\textwidth}
\begin{tabular}{rrrc@{\hspace*{7pt}}cccccccc@{\hspace*{7pt}}ccccccc}
\toprule
\multicolumn{ 3}{c}{} &  & \multicolumn{ 7}{c}{plain CSCP-$r$ (LP)} &  & \multicolumn{ 7}{c}{full CSCP-$r$ (LP)} \\ \cline{5-11} \cline{13-19}
\multicolumn{ 3}{c}{\uppp{instance}} &  &  & \multicolumn{ 7}{c}{gap(\%)} &  &  & \multicolumn{ 4}{c}{gap(\%)} \\ \cline{1-3} \cline{6-11} \cline{14-19}
$n$ & $p$ & \# &  & {LB} & $+$\eqref{eq:CSCP3} & $+$\eqref{eq:ineq3}$^{a}$ & $+$\eqref{eq:ineq3}$^{b}$ & $+$\eqref{eq:ineq5}$^{a}$ & $+$\eqref{eq:ineq5}$^{b}$ & $+$\eqref{ourSSineq} &  & {LB} & $-$\eqref{eq:CSCP3} & $-$\eqref{eq:ineq3}$^{c}$ & $-$\eqref{eq:ineq3}$^{d}$ & $-$\eqref{eq:ineq5}$^{c}$ & $-$\eqref{eq:ineq5}$^{d}$ & $-$\eqref{ourSSineq} \\ 
\cline{1-3} \cline{5-11} \cline{13-19}
50 & 5 & 10 &  & 4.46 & 18.50 & 0.00 & 0.00 & 0.00 & 0.00 & 2.47 &  & 5.28 & -7.20 & 0.00 & 0.00 & 0.00 & 0.00 & -0.03 \\ 
100 & 10 & 10 &  & 8.82 & 14.46 & 0.00 & 0.00 & 0.00 & 0.00 & 3.99 &  & 10.10 & -4.91 & -0.01 & -0.01 & 0.00 & 0.00 & -0.02 \\ 
150 & 15 & 10 &  & 12.92 & 21.25 & 0.00 & 0.00 & 0.00 & 0.00 & 5.98 &  & 15.67 & -7.09 & -0.01 & -0.01 & 0.00 & 0.00 & -0.04 \\ 
200 & 20 & 10 &  & 17.58 & 16.00 & 0.00 & 0.00 & 0.00 & 0.00 & 3.61 &  & 20.41 & -5.81 & 0.00 & 0.00 & 0.00 & 0.00 & -0.06 \\ 
\cline{1-3} \cline{5-11} \cline{13-19}
50 & 12 & 10 &  & 10.70 & 10.01 & 0.00 & 0.00 & 0.49 & 0.90 & 5.47 &  & 12.00 & -1.09 & -0.04 & -0.05 & -0.48 & -0.48 & -0.91 \\ 
100 & 25 & 10 &  & 22.05 & 10.13 & 0.00 & 0.00 & 0.35 & 1.84 & 5.49 &  & 25.05 & -1.12 & -0.01 & -0.03 & -1.90 & -1.90 & -0.56 \\ 
150 & 37 & 10 &  & 31.64 & 15.65 & 0.00 & 0.00 & 0.61 & 2.24 & 8.13 &  & 37.95 & -2.04 & -0.02 & -0.03 & -2.20 & -2.20 & -0.64 \\ 
200 & 50 & 10 &  & 43.95 & 11.93 & 0.00 & 0.00 & 0.97 & 1.85 & 6.41 &  & 51.02 & -1.57 & -0.02 & -0.04 & -1.92 & -1.92 & -0.88 \\ 
\cline{1-3} \cline{5-11} \cline{13-19}
50 & 16 & 10 &  & 14.07 & 10.22 & 0.00 & 0.00 & 1.35 & 1.42 & 6.15 &  & 16.03 & -0.21 & -0.03 & -0.03 & -1.31 & -1.43 & -1.27 \\ 
100 & 33 & 10 &  & 28.61 & 10.81 & 0.00 & 0.00 & 1.81 & 2.43 & 6.85 &  & 32.88 & -0.64 & -0.04 & -0.08 & -1.34 & -1.40 & -1.43 \\ 
150 & 50 & 10 &  & 43.07 & 10.81 & 0.00 & 0.00 & 2.21 & 3.19 & 7.06 &  & 50.50 & -0.59 & -0.06 & -0.07 & -2.14 & -2.22 & -2.06 \\ 
200 & 66 & 10 &  & 57.01 & 12.11 & 0.00 & 0.00 & 2.66 & 3.74 & 8.33 &  & 68.16 & -0.57 & -0.06 & -0.10 & -2.34 & -2.43 & -1.83 \\ 
\cline{1-3} \cline{5-11} \cline{13-19}
50 & 20 & 10 &  & 17.83 & 4.23 & 0.00 & 0.00 & 1.70 & 2.00 & 6.21 &  & 19.58 & -0.06 & -0.04 & -0.18 & -0.86 & -0.94 & -3.08 \\ 
100 & 40 & 10 &  & 35.29 & 4.48 & 0.00 & 0.00 & 2.06 & 2.65 & 6.24 &  & 39.07 & -0.15 & -0.06 & -0.14 & -1.01 & -1.27 & -3.06 \\ 
150 & 60 & 10 &  & 51.68 & 7.00 & 0.00 & 0.00 & 2.41 & 3.21 & 6.61 &  & 58.72 & -0.13 & -0.13 & -0.17 & -1.26 & -1.49 & -2.84 \\ 
200 & 80 & 10 &  & 70.32 & 7.07 & 0.00 & 0.00 & 3.02 & 3.31 & 7.12 &  & 80.30 & -0.21 & -0.10 & -0.18 & -1.27 & -1.56 & -3.07 \\ 
\cline{1-3} \cline{5-11} \cline{13-19}
\multicolumn{ 3}{c}{average} &  & - & 11.54 & 0.00 & 0.00 & 1.23 & 1.80 & 6.01 &  & - & -2.09 & -0.04 & -0.07 & -1.13 & -1.20 & -1.36 \\ 
\bottomrule
\multicolumn{17}{l}{\small $^{a}$ $|S| \in \{1,2\}$. $^{b}$ $|S| \in \{1,2,3\}$. $^{c}$ $|S| \in \{2,3\}$. $^{d}$ $|S| \in \{1,2,3\}$.}
\end{tabular}
\end{adjustbox}
}
\label{tab:evalLBs}
\end{table}

\noindent
\textbf{Impact of the symmetry-breaking inequalities.}
The impact of inequalities \eqref{eq:ineq1}, \eqref{eq:ineq2}, and \eqref{eq:ineq4} is not shown in Table \ref{tab:evalLBs}, because they do not enhance the linear relaxation but rather help to reduce the search space. We thus evaluated their impact by computing the number of nodes explored in the subproblems. For this experiment, we selected the instances from S1 that were not proved to be infeasible after the solution of the LP relaxation of the full CSCP-$r$ (68 instances in total) and solved them with our models, terminating upon proven infeasibility or when a time limit of 600 seconds was reached. Table \ref{tab:nb-nodes} reports the average number of nodes explored by GUROBI when solving the plain and full CSCP-$r$ formulations (columns 4 and 8, respectively); the formulations obtained by adding inequalities \eqref{eq:ineq1}, \eqref{eq:ineq2}, or~\eqref{eq:ineq4} one at a time to the plain CSCP-$r$ (columns 5--7); and the formulations obtained by removing these inequalities one at a time from the full CSCP-$r$ (columns 9--11).
We observe a significant decrease in the number of nodes explored when using the full CPCP-$r$ instead of the plain version. 
The contribution of the three inequalities is not always relevant when they are added one at a time to the plain formulation. 
One instance in particular (with $n$=100 and $p$=25) could not be proven infeasible within the time limit, either by the plain CSCP-$r$ or by adding \eqref{eq:ineq2}. However, the number of nodes explored increased significantly with the use of inequalities \eqref{eq:ineq2}. 
Their removal from the full formulation is detrimental, since it significantly increases the number of nodes explored.

\begin{table}[htbp]
\centering\setlength{\tabcolsep}{6.5pt}
\renewcommand{\arraystretch}{1.3}
\centering
\caption{Number of nodes explored to prove infeasibility with CSCP-$r$ formulations}
{
\begin{adjustbox}{max width=\textwidth}
\begin{tabular}{rrr@{\hspace*{5pt}}lrrrr@{\hspace*{5pt}}lrrrr}
\toprule
\multicolumn{ 3}{c}{instance} & & \multicolumn{4}{c}{explored nodes} &  & \multicolumn{4}{c}{explored nodes}  \\ 
\cline{ 1- 3} \cline{5-8} \cline{10-13}
$n$ & $p$ & \# &  & {plain} & {$+$\eqref{eq:ineq1}} & {$+$\eqref{eq:ineq2}} & {$+$\eqref{eq:ineq4}} &  & {full} & {$-$\eqref{eq:ineq1}} & {$-$\eqref{eq:ineq2}} & {$-$\eqref{eq:ineq4}} \\ 
\cline{ 1- 3} \cline{5-8} \cline{10-13}
50 & 5 & 2 &  & 22.00 & 10.00 & 64.50 & 52.00 &  & 1.00 & 1.00 & 1.00 & 1.00 \\ 
100 & 10 & 5 &  & 4,024.20 & 607.60 & 1,737.80 & 3,402.40 &  & 310.00 & 244.20 & 305.20 & 191.60 \\ 
150 & 15 & 3 &  & 17,512.33 & 4,820.33 & 10,270.00 & 9,460.33 &  & 1,482.00 & 9,927.00 & 1,648.67 & 1,299.00 \\ 
200 & 20 & 3 &  & 3,891.00 & 1,071.00 & 4,100.67 & 2,656.00 &  & 285.33 & 200.33 & 89.33 & 1.00 \\
\cline{ 1- 3} \cline{5-8} \cline{10-13}
50 & 12 & 4 &  & 996.50 & 112.00 & 759.25 & 126.00 &  & 1.00 & 1.75 & 1.00 & 1.00 \\ 
100 & 25 & 5 &  & 46,141.80 & 12,563.00 & 111,176.80 & 33,091.60 &  & 6,017.00 & 28,510.40 & 19,308.80 & 11,141.80 \\ 
150 & 37 & 0 &  & - & - & - & - &  & - & - & - & - \\ 
200 & 50 & 2 &  & 1.00 & 1.00 & 1.00 & 1.00 &  & 1.00 & 1.00 & 1.00 & 1.00 \\ 
\cline{ 1- 3} \cline{5-8} \cline{10-13}
50 & 16 & 5 &  & 1,706.00 & 205.20 & 2,237.20 & 3,183.40 &  & 1.00 & 3.20 & 1.00 & 1.00 \\ 
100 & 33 & 6 &  & 625.00 & 637.33 & 695.50 & 494.50 &  & 1.00 & 18.83 & 1.00 & 1.00 \\ 
150 & 50 & 4 &  & 2,110.25 & 1,002.50 & 2,832.00 & 3,606.50 &  & 144.50 & 162.00 & 138.00 & 376.00 \\ 
200 & 66 & 2 &  & 420.50 & 521.00 & 263.50 & 266.00 &  & 1.00 & 1.00 & 1.00 & 1.00 \\ 
\cline{ 1- 3} \cline{5-8} \cline{10-13}
50 & 20 & 9 &  & 2,515.78 & 2,767.89 & 2,000.67 & 1,375.11 &  & 155.44 & 71.67 & 474.56 & 138.56 \\ 
100 & 40 & 8 &  & 1,637.63 & 397.75 & 340.25 & 1,834.88 &  & 10.88 & 107.25 & 19.00 & 65.50 \\ 
150 & 60 & 7 &  & 691.71 & 342.43 & 410.14 & 670.71 &  & 87.00 & 90.00 & 85.00 & 85.57 \\ 
200 & 80 & 3 &  & 380.67 & 384.00 & 365.67 & 363.00 &  & 1.00 & 1.00 & 1.00 & 1.00 \\ 
\cline{ 1- 3} \cline{5-8} \cline{10-13}
\multicolumn{ 3}{c}{average} &  & 5,511.76 & 1,696.20 & 9,150.33 & 4,038.90 &  & 566.61 & 2,622.71 & 1,471.70 & 887.07 \\
\bottomrule
\end{tabular}
\end{adjustbox}
}
\label{tab:nb-nodes}
\vspace{-.5cm}
\end{table}

\subsection{Performance of the search algorithms for the CPCP} \label{sec:search-results}

Based on the results of Section \ref{sec:eval-formulations}, we decided to integrate the full CSCP-$r$, composed of  \eqref{eq:CSCPinit}--\eqref{eq:ineq5} and \eqref{ourSSineq}, into the search algorithms presented in Section \ref{sec:method}. {We now evaluate the performance of the resulting search approaches for the CPCP. We set the initial search interval for the radius $r$ to $[0,\text{UB}_\text{best}]$, where UB$_\text{best}$ is the best upper bound found by the heuristic of \citet{QuevedoOrozco2015} and by the ILS that we described at the end of Section \ref{sec:method}.}
Table \ref{tab:comp-lit} compares the results obtained by the proposed algorithms with those of recent state-of-the-art heuristic and exact approaches, on the classical instances from the literature, namely: 
\begin{itemize}[leftmargin=2cm]
\setlength\itemsep{.035cm}
\item[SPS --] the multi-exchange neighborhood search algorithm of \citet{Scaparraetal2004};
\item[OP --] the binary search of \citet{Ozsoy2006} that solves plain CSCP models at each iteration;
\item[ADF --] the Lagrangian relaxation method of \citet{AlbaredaSambola2010}; 
\item[QR --] the iterated greedy local search metaheuristic of \citet{QuevedoOrozco2015}, including supplementary results available from a personal communication.
\end{itemize}
{For our search algorithms,} we tested sequential search (SS), L-layered search with $L$=2, 3, or 4 (L2, L3, L4), and binary search (BS). 
Each row in the table reports the results for a subset of instances.  For each algorithm and subset, we report the number of proven optimal solutions ``\#opt'', the average gap ``gap(\%)'', and the average computing time in seconds ``T(s)''. For the exact algorithms (all but SPS and QR),  gap(\%) is evaluated as $100\cdot(\text{UB}-\text{LB})/\text{UB}$, where $\text{UB}$ and $\text{LB}$ are the upper and lower bounds produced by the given algorithm, and \#opt is the number of times for which LB=UB. For the heuristics (SPS and QR), gap(\%) and \#opt are computed by considering the best known LB values.
{The results for SPS, OP, ADF, and QR were taken from \citet{QuevedoOrozco2015}, who ran the algorithms on a 2.0GHz AMD Opteron processor with 32GB of RAM. As they did not solve the S3 instances with float distances, the ADF results for these instances were collected from \citet{AlbaredaSambola2010}, who used a  2.39GHz Intel Pentium 4 processor with 512MB of RAM.
}

{From Table \ref{tab:comp-lit}, we can observe that the results of SS, L2, L3, L4, and BS on the set S1 are better than those of the previous literature for the majority of the instances. 
The ten instances with $n=200$ and $p=20$ are solved to optimality only by OP, SS, and L3, while those with $n=200$ and $p=50$ are solved only by SS. Instances from S2 and S3-int are solved by all the exact methods. Of the six instances from S3-float, five are solved to proven optimality by L2, L3, L4 and BS, whereas SS can solve only two. The five instances from S4, which involve more than 3000 vertices, are not solved to proven optimality by any of the algorithms. This is also the only set for which the average gap of our L-layered search algorithms is large, although it is still smaller than that of SPS, ADF, and QR (that for OP is not computed because the memory exceeded the limit imposed by the authors, and no upper bound could be found). The average CPU time of SS, L2, L3, L4, and BS is smaller than that of OP and ADF, due to more generous time limits for the latter approaches. Despite this difference, SS, L2, L3, L4, and BS are all able to produce better results. Between them, our five algorithms were able prove optimality for 179 of the 185 benchmark instances, including 13 open cases.}

\begin{landscape}
\begin{table}[t]
\centering
\setlength{\tabcolsep}{3.15pt}
\renewcommand{\arraystretch}{1.85}
\vspace*{0.4cm}
\caption{Evaluation of existing (SPS, OP, ADF, QR) and new (SS, L2, L3, L4, BS) exact CPCP algorithms on classical benchmark sets.}
\scalebox{.78}
{
\hspace{-.5cm}
\begin{tabular}{@{\hspace*{0.0pt}}l@{\hspace*{0.2pt}}rrrrrr@{\hspace*{1.8pt}}r@{\hspace*{1.8pt}}rrrrrrrrrrrrrrrr@{\hspace*{0.0pt}}r@{\hspace*{0.0pt}}rr@{\hspace*{2pt}}rrrrrrr@{\hspace*{0.0pt}}}
\toprule
\multicolumn{ 4}{c}{instance} & \multicolumn{ 9}{c}{\#opt} &  &  & \multicolumn{ 9}{c}{gap(\%)} &  & \multicolumn{ 9}{c}{T(s)} \\ 
\cline{1-4} \cline{6-14} \cline{16-24} \cline{26-34}
set & $n$ & $p$ & \# & & SPS & OP & ADF & QR & SS & L2 & L3 & L4 & BS & & SPS & OP & ADF & QR & SS & L2 & L3 & L4 & BS & & SPS & OP & ADF & QR & SS & L2 & L3 & L4 & BS \\ 
\cline{1-4} \cline{6-14} \cline{16-24} \cline{26-34}
S1 & 50 & 5 & 10 &  & 7 & 10 & 10 & 10 & 10 & 10 & 10 & 10 & 10 &  & 1.88 & 0 & 0 & 0 & 0 & 0 & 0 & 0 & 0 &  & 1.2 & 2.9 & 1.4 & 0.2 & 0.8 & 1.2 & 0.9 & 0.7 & 0.4 \\ 
 &  & 12 & 10 &  & 2 & 10 & 10 & 3 & 10 & 10 & 10 & 10 & 10 &  & 5.39 & 0 & 0 & 6.64 & 0 & 0 & 0 & 0 & 0 &  & 2.4 & 5.5 & 9.7 & 0.4 & 0.9 & 1.4 & 1.5 & 1.1 & 1.7 \\ 
 &  & 16 & 10 &  & 3 & 10 & 10 & 2 & 10 & 10 & 10 & 10 & 10 &  & 2.87 & 0 & 0 & 8.97 & 0 & 0 & 0 & 0 & 0 &  & 2.8 & 21.3 & 82.8 & 0.6 & 0.7 & 0.8 & 1.4 & 1.3 & 0.9 \\ 
 &  & 20 & 10 &  & 2 & 10 & 8 & 0 & 10 & 10 & 10 & 10 & 10 &  & 3.59 & 0 & 2.14 & 21.28 & 0 & 0 & 0 & 0 & 0 &  & 1.9 & 86.5 & 1103.1 & 0.7 & 2.6 & 4.8 & 4.8 & 3.9 & 2.9 \\ \hline
 & 100 & 10 & 10 &  & 0 & 10 & 10 & 8 & 10 & 10 & 10 & 10 & 10 &  & 8.01 & 0 & 0 & 1.30 & 0 & 0 & 0 & 0 & 0 &  & 9.5 & 62.4 & 59.4 & 0.6 & 4.0 & 6.8 & 6.4 & 4.8 & 4.2 \\ 
 &  & 25 & 10 &  & 1 & 9 & 8 & 3 & 10 & 10 & 10 & 10 & 10 &  & 12.97 & 1.33 & 3.21 & 9.08 & 0 & 0 & 0 & 0 & 0 &  & 9.7 & 423.4 & 975.6 & 1.4 & 29.1 & 22.1 & 26.0 & 35.0 & 35.1 \\ 
 &  & 33 & 10 &  & 0 & 9 & 7 & 2 & 10 & 10 & 10 & 10 & 10 &  & 10.34 & 1.67 & 2.32 & 9.46 & 0 & 0 & 0 & 0 & 0 &  & 15.6 & 377.6 & 1380.4 & 2.1 & 4.7 & 8.4 & 9.9 & 8.3 & 6.8 \\ 
 &  & 40 & 10 &  & 1 & 10 & 6 & 0 & 10 & 10 & 10 & 10 & 10 &  & 9.89 & 0 & 3.93 & 35.82 & 0 & 0 & 0 & 0 & 0 &  & 30.8 & 191.0 & 2021.3 & 2.4 & 3.4 & 3.8 & 6.6 & 8.7 & 5.8 \\ \hline
 & 150 & 15 & 10 &  & 0 & 10 & 10 & 4 & 10 & 10 & 10 & 10 & 10 &  & 13.74 & 0 & 0 & 3.54 & 0 & 0 & 0 & 0 & 0 &  & 26.6 & 154.4 & 165.1 & 1.0 & 11.2 & 16.7 & 30.0 & 30.7 & 14.4 \\ 
 &  & 37 & 10 &  & 0 & 9 & 9 & 0 & 10 & 10 & 10 & 10 & 10 &  & 18.14 & 1.67 & 0.83 & 12.31 & 0 & 0 & 0 & 0 & 0 &  & 46.9 & 386.0 & 532.5 & 2.6 & 6.5 & 16.1 & 34.1 & 37.6 & 12.4 \\ 
 &  & 50 & 10 &  & 0 & 9 & 7 & 0 & 10 & 10 & 10 & 10 & 10 &  & 16.1 & 1.67 & 3.50 & 22.84 & 0 & 0 & 0 & 0 & 0 &  & 34.6 & 560.9 & 1834.6 & 5.5 & 1.9 & 4.7 & 6.0 & 10.8 & 4.0 \\ 
 &  & 60 & 10 &  & 0 & 9 & 6 & 0 & 10 & 10 & 10 & 10 & 10 &  & 16.2 & 3.57 & 4.21 & 38.47 & 0 & 0 & 0 & 0 & 0 &  & 51.5 & 914.8 & 1908.3 & 7.2 & 4.8 & 10.1 & 12.3 & 22.8 & 11.6 \\ \cline{1-4} \cline{6-14} \cline{16-24} \cline{26-34}
 & 200 & 20 & 10 &  & 0 & 10 & 9 & 3 & 10 & 9 & 10 & 9 & 9 &  & 18.8 & 0 & 0.67 & 5.26 & 0 & 1.25 & 0 & 2.50 & 1.88 &  & 63.6 & 237.5 & 625.9 & 1.6 & 31.6 & 75.6 & 66.8 & 72.5 & 68.4 \\ 
 &  & 50 & 10 &  & 0 & 6 & 4 & 0 & 10 & 9 & 8 & 9 & 9 &  & 21.08 & 4.22 & 11.20 & 22.24 & 0 & 4.62 & 4.00 & 1.00 & 2.50 &  & 46.8 & 1145.3 & 2469.8 & 6.5 & 107.3 & 110.5 & 158.9 & 186.0 & 163.1 \\ 
 &  & 66 & 10 &  & 0 & 9 & 4 & 0 & 10 & 10 & 10 & 10 & 10 &  & 23.13 & 2.22 & 8.76 & 31.23 & 0 & 0 & 0 & 0 & 0 &  & 55.3 & 846.2 & 2163.5 & 9.8 & 9.3 & 45.9 & 26.4 & 34.1 & 27.9 \\ 
 &  & 80 & 10 &  & 0 & 7 & 5 & 0 & 10 & 10 & 10 & 9 & 10 &  & 23.91 & 7.17 & 9.05 & 36.53 & 0 & 0 & 0 & 4.74 & 0 &  & 94.4 & 560.3 & 1800.6 & 11.9 & 36.9 & 33.6 & 69.6 & 69.0 & 64.1 \\ 
 \cline{1-4} \cline{6-14} \cline{16-24} \cline{26-34}
\multicolumn{3}{l}{S2} & 8 &  & 0 & 8 & 8 & 1 & 8 & 8 & 8 & 8 & 8 &  & 6.21 & 0 & 0 & 3.25 & 0 & 0 & 0 & 0 & 0 &  & 27.3 & 152.7 & 219.4 & 0.8 & 51.9 & 64.0 & 99.3 & 57.6 & 72.0 \\ 
\multicolumn{3}{l}{S3-int} & 6 &  & 0 & 6 & 6 & 0 & 6 & 6 & 6 & 6 & 6 &  & 23.96 & 0 & 0 & 3.57 & 0 & 0 & 0 & 0 & 0 &  & 285.4 & 319.3 & 256.6 & 2.4 & 143.4 & 65.6 & 84.5 & 61.5 & 57.1 \\ 
\multicolumn{3}{l}{S3-float} & 6 &  & -- & -- & 4 & -- & 2 & 5 & 5 & 5 & 5 &  & -- & -- & 0.55 & -- & 23.38 & 0.18 & 0.04 & 0.02 & 0.00 &  & -- & -- & 1201.5 & -- & 555.5 & 242.6 & 190.2 & 150.1 & 150.0 \\ 
\multicolumn{3}{l}{S4}  & 5 &  & 0 & 0 & 0 & 0 & 0 & 0 & 0 & 0 & 0 &  & 75.61 & -- & 60.94 & 53.14 & 62.38 & 28.28 & 15.35 & 11.95 & 14.6 &  & 1877.6 & -- & 21600.0 & 350.4 & 600.0 & 600.0 & 600.0 & 600.0 & 600.0 \\ 
\cline{1-4} \cline{6-14} \cline{16-24} \cline{26-34}
\multicolumn{3}{c}{avg/sum} & 185 & & 16 & 161 & 141 & 36 & 176 & 177 & 177 & 176 & 177 & \multicolumn{1}{l}{} & 16.41 & 1.31 & 5.57 & 17.10 & 4.29 & 1.72 & 0.97 & 1.01 & 0.95 & \multicolumn{1}{l}{} & 141.3 & 343.1 & 2020.6 & 21.5 & 80.3 & 66.7 & 71.8 & 69.8 & 65.1 \\ 
\bottomrule
\end{tabular}
\label{tab:comp-lit}
}
\end{table}
\end{landscape}

Overall, it seems that SS performs especially well for small instances with integer distances, while BS behaves better on large instances with floating-point distances. The L-layered search algorithms have a performance that is quite close to the best results on both types of instances. This analysis, however, is restricted because of the limited variety of the existing instances, with few test cases involving large and/or floating-point distances. To fill this experimental gap, we conducted additional experiments on the new KIV test set, reported in Tables \ref{tab:CPCP-results} and \ref{tab:newinstfloat}, with either integer or floating point distances. These two tables have the same format as Table \ref{tab:comp-lit}.

\begin{table}[htbp]
\setlength{\tabcolsep}{3.8pt}
\renewcommand{\arraystretch}{1.3}
\centering
\caption{Performance of search algorithms on KIV instances -- integer distances.}
\begin{adjustbox}{max width=\textwidth}
{
\begin{tabular}{Hrrrlrrrrrlrrrrrlrrrrr}
\toprule
\multicolumn{ 4}{c}{{instance}} & \multicolumn{1}{c}{{}} & \multicolumn{ 5}{c}{{\#opt}} & \multicolumn{1}{c}{{}} & \multicolumn{ 5}{c}{{gap(\%)}} & \multicolumn{1}{c}{{}} & \multicolumn{ 5}{c}{{T(s)}} \\ \cline{1-4} \cline{6-10} \cline{12-16} \cline{18-22}
{set} & {$n$} & {$p$} & {\#} & \multicolumn{1}{r}{{}} & {SS} & {L2} & {L3} & {L4} & {BS} & \multicolumn{1}{r}{{}} & {SS} & {L2} & {L3} & {L4} & {BS} & \multicolumn{1}{r}{{}} & {SS} & {L2} & {L3} & {L4} & {BS} \\ \cline{1-4} \cline{6-10} \cline{12-16} \cline{18-22}
KIV & 300 & 30 & 5 &  & 5 & 5 & 5 & 4 & 5 &  & 0.00 & 0.00 & 0.00 & 0.00 & 4.17 &  & 149.85 & 115.83 & 160.70 & 131.49 & 244.88 \\ 
 &  & 42 & 5 &  & 3 & 3 & 3 & 3 & 3 &  & 8.83 & 6.00 & 2.22 & 3.33 & 4.44 &  & 266.02 & 265.96 & 272.29 & 278.29 & 248.85 \\ 
 &  & 75 & 5 &  & 5 & 5 & 5 & 5 & 5 &  & 0.00 & 0.00 & 0.00 & 0.00 & 0.00 &  & 64.81 & 35.48 & 44.22 & 25.93 & 114.55 \\ 
 &  & 100 & 5 &  & 5 & 5 & 5 & 5 & 5 &  & 0.00 & 0.00 & 0.00 & 0.00 & 0.00 &  & 7.81 & 7.31 & 8.40 & 10.54 & 8.99 \\ \cline{1-4} \cline{6-10} \cline{12-16} \cline{18-22}
 & 500 & 50 & 5 &  & 3 & 3 & 3 & 3 & 3 &  & 4.63 & 8.20 & 5.17 & 6.45 & 10.80 &  & 384.85 & 324.63 & 297.72 & 294.51 & 268.99 \\ 
 &  & 71 & 5 &  & 3 & 3 & 3 & 3 & 3 &  & 8.00 & 14.00 & 10.00 & 2.35 & 2.35 &  & 291.11 & 260.54 & 294.79 & 271.94 & 287.95 \\ 
 &  & 125 & 5 &  & 3 & 3 & 3 & 3 & 3 &  & 17.62 & 6.42 & 4.29 & 5.33 & 22.09 &  & 274.87 & 257.56 & 270.51 & 277.59 & 264.58 \\ 
 &  & 166 & 5 &  & 4 & 3 & 2 & 4 & 4 &  & 10.00 & 16.50 & 18.75 & 2.86 & 19.17 &  & 315.53 & 360.84 & 378.13 & 299.34 & 343.37 \\ \cline{1-4} \cline{6-10} \cline{12-16} \cline{18-22}
 & 1000 & 100 & 5 &  & 1 & 1 & 0 & 0 & 0 &  & 19.06 & 5.00 & 19.33 & 17.91 & 14.61 &  & 520.71 & 598.12 & 600.00 & 600.00 & 600.00 \\ 
 &  & 142 & 5 &  & 0 & 0 & 0 & 0 & 0 &  & 29.12 & 42.41 & 41.82 & 11.11 & 11.50 &  & 600.00 & 600.00 & 600.00 & 600.00 & 600.00 \\ 
 &  & 250 & 5 &  & 2 & 2 & 2 & 2 & 2 &  & 29.17 & 17.51 & 9.11 & 8.19 & 57.56 &  & 408.05 & 423.35 & 449.11 & 433.46 & 410.31 \\ 
 &  & 333 & 5 &  & 1 & 1 & 1 & 1 & 1 &  & 44.05 & 16.00 & 20.29 & 31.29 & 7.69 &  & 490.49 & 491.41 & 500.35 & 538.24 & 497.69 \\ \cline{1-4} \cline{6-10} \cline{12-16} \cline{18-22}
 & 2000 & 200 & 5 &  & 0 & 0 & 0 & 0 & 0 &  & 31.24 & 38.88 & 35.06 & 54.16 & 45.04 &  & 600.00 & 600.00 & 600.00 & 600.00 & 600.00 \\ 
 &  & 285 & 5 &  & 0 & 0 & 0 & 0 & 0 &  & 38.69 & 40.36 & 65.04 & 54.22 & 34.60 &  & 600.00 & 600.00 & 600.00 & 600.00 & 600.00 \\ 
 &  & 500 & 5 &  & 3 & 3 & 3 & 3 & 1 &  & 22.54 & 8.00 & 7.79 & 38.46 & 3.08 &  & 421.82 & 443.65 & 505.54 & 596.94 & 419.01 \\ 
 &  & 666 & 5 &  & 4 & 4 & 4 & 4 & 1 &  & 11.72 & 4.00 & 1.54 & 51.59 & 8.00 &  & 243.88 & 243.29 & 308.17 & 539.54 & 236.96 \\ \cline{1-4} \cline{6-10} \cline{12-16} \cline{18-22}
 & 3000 & 300 & 5 &  & 0 & 0 & 0 & 0 & 0 &  & 33.77 & 41.13 & 47.45 & 54.42 & 45.58 &  & 600.00 & 600.00 & 600.00 & 600.00 & 600.00 \\ 
 &  & 428 & 5 &  & 0 & 0 & 0 & 0 & 0 &  & 42.78 & 54.88 & 65.40 & 54.11 & 45.09 &  & 600.00 & 600.00 & 600.00 & 600.00 & 600.00 \\ 
 &  & 750 & 5 &  & 2 & 1 & 0 & 2 & 0 &  & 35.04 & 12.00 & 12.82 & 70.83 & 10.97 &  & 566.33 & 583.55 & 600.00 & 600.00 & 567.66 \\ 
 &  & 1000 & 5 &  & 4 & 4 & 2 & 4 & 4 &  & 12.14 & 5.33 & 13.50 & 5.00 & 8.57 &  & 340.07 & 350.55 & 575.66 & 305.95 & 383.53 \\ \cline{1-4} \cline{6-10} \cline{12-16} \cline{18-22}
\multicolumn{ 3}{c}{avg/sum} & 100 &  & {48} & 46 & 41 & 46 & 40 &  & 19.92 & {16.83} & 18.98 & 23.58 & 17.77 &  & {387.31} & 388.10 & 413.28 & 410.19 & 394.87 \\ \bottomrule
\end{tabular}
}
\end{adjustbox}
\label{tab:CPCP-results}
\vspace{-.3cm}
\end{table}

\begin{table}[htbp]
\setlength{\tabcolsep}{3.4pt}
\renewcommand{\arraystretch}{1.45}
\centering
\caption{Performance of search algorithms on KIV instances -- floating-point distances.}
\scalebox{.87}
{
\begin{tabular}{Hrrrlrrrrrlrrrrrlrrrrr}
\toprule
\multicolumn{ 4}{c}{instance} & \multicolumn{1}{r}{} & \multicolumn{ 5}{c}{\#opt} & \multicolumn{1}{r}{} & \multicolumn{ 5}{c}{gap(\%)} & \multicolumn{1}{r}{} & \multicolumn{ 5}{c}{T(s)} \\ \cline{1-4} \cline{6-10} \cline{12-16} \cline{18-22}
set & $n$ & $p$ & \# & \multicolumn{1}{r}{} & SS & L2 & L3 & L4 & BS & \multicolumn{1}{r}{} & SS & L2 & L3 & L4 & BS & \multicolumn{1}{r}{} & SS & L2 & L3 & L4 & BS \\ \cline{1-4} \cline{6-10} \cline{12-16} \cline{18-22}
KIV-float & 50 & 5 & 5 &  & 5 & 5 & 5 & 5 & 5 &  & 0.00 & 0.00 & 0.00 & 0.00 & 0.00 &  & 5.32 & 1.28 & 1.05 & 1.39 & 0.99 \\ 
 &  & 7 & 5 &  & 5 & 5 & 5 & 5 & 5 &  & 0.00 & 0.00 & 0.00 & 0.00 & 0.00 &  & 6.10 & 2.07 & 2.34 & 1.42 & 1.30 \\ 
 &  & 12 & 5 &  & 5 & 5 & 5 & 5 & 3 &  & 0.00 & 0.00 & 0.00 & 0.00 & 27.12 &  & 1.28 & 0.60 & 0.90 & 0.55 & 1.07 \\ 
 &  & 16 & 5 &  & 5 & 5 & 5 & 5 & 5 &  & 0.00 & 0.00 & 0.00 & 0.00 & 0.00 &  & 0.60 & 0.33 & 0.34 & 0.37 & 0.29 \\ \cline{1-4} \cline{6-10} \cline{12-16} \cline{18-22}
 & 100 & 10 & 5 &  & 5 & 5 & 5 & 5 & 5 &  & 0.00 & 0.00 & 0.00 & 0.00 & 0.00 &  & 30.17 & 17.19 & 13.49 & 15.93 & 14.83 \\ 
 &  & 14 & 5 &  & 5 & 5 & 5 & 5 & 5 &  & 0.00 & 0.00 & 0.00 & 0.00 & 0.00 &  & 15.60 & 10.25 & 8.58 & 15.35 & 10.46 \\ 
 &  & 25 & 5 &  & 4 & 4 & 4 & 4 & 4 &  & 6.24 & 1.06 & 0.61 & 0.14 & 0.57 &  & 158.07 & 220.95 & 166.26 & 166.96 & 167.92 \\ 
 &  & 33 & 5 &  & 5 & 5 & 5 & 5 & 5 &  & 0.00 & 0.00 & 0.00 & 0.00 & 0.00 &  & 42.58 & 48.79 & 70.50 & 31.30 & 36.28 \\ \cline{1-4} \cline{6-10} \cline{12-16} \cline{18-22}
 & 150 & 15 & 5 &  & 5 & 5 & 5 & 5 & 5 &  & 0.00 & 0.00 & 0.00 & 0.00 & 0.00 &  & 32.23 & 13.05 & 21.64 & 20.45 & 12.58 \\ 
 &  & 21 & 5 &  & 5 & 4 & 4 & 5 & 5 &  & 0.00 & 1.17 & 0.10 & 0.00 & 0.00 &  & 166.67 & 172.82 & 185.70 & 77.61 & 126.11 \\ 
 &  & 37 & 5 &  & 4 & 4 & 4 & 5 & 4 &  & 5.58 & 0.05 & 0.15 & 0.00 & 0.05 &  & 237.82 & 247.20 & 260.61 & 225.36 & 211.39 \\ 
 &  & 50 & 5 &  & 4 & 4 & 4 & 4 & 4 &  & 5.24 & 0.30 & 2.88 & 0.62 & 0.87 &  & 156.58 & 153.83 & 159.03 & 145.42 & 146.85 \\ \cline{1-4} \cline{6-10} \cline{12-16} \cline{18-22}
 & 200 & 20 & 5 &  & 5 & 5 & 4 & 5 & 4 &  & 0.00 & 0.00 & 0.29 & 0.00 & 0.20 &  & 169.19 & 87.50 & 166.74 & 169.80 & 184.77 \\ 
 &  & 28 & 5 &  & 5 & 5 & 5 & 5 & 5 &  & 0.00 & 0.00 & 0.00 & 0.00 & 0.00 &  & 71.36 & 49.19 & 41.97 & 29.46 & 49.32 \\ 
 &  & 50 & 5 &  & 4 & 4 & 5 & 5 & 4 &  & 5.96 & 0.75 & 0.00 & 0.00 & 0.10 &  & 245.93 & 282.82 & 116.00 & 210.95 & 226.99 \\ 
 &  & 66 & 5 &  & 4 & 4 & 4 & 4 & 4 &  & 4.40 & 5.43 & 0.52 & 1.34 & 0.88 &  & 173.39 & 180.51 & 168.94 & 161.10 & 233.10 \\ \cline{1-4} \cline{6-10} \cline{12-16} \cline{18-22}
KIV-float & 300 & 30 & 5 &  & 2 & 2 & 2 & 3 & 2 &  & 11.80 & 4.42 & 2.88 & 0.89 & 6.92 &  & 487.59 & 507.61 & 475.35 & 416.42 & 476.94 \\ 
 &  & 42 & 5 &  & 0 & 1 & 0 & 1 & 0 &  & 23.23 & 12.98 & 10.62 & 13.67 & 22.82 &  & 600.00 & 598.77 & 600.00 & 565.02 & 600.00 \\ 
 &  & 75 & 5 &  & 1 & 1 & 1 & 1 & 1 &  & 23.69 & 9.26 & 17.61 & 22.68 & 5.46 &  & 516.02 & 525.55 & 575.60 & 503.92 & 581.75 \\ 
 &  & 100 & 5 &  & 4 & 4 & 4 & 4 & 4 &  & 5.56 & 0.23 & 1.14 & 0.80 & 0.17 &  & 192.91 & 155.79 & 166.97 & 159.29 & 172.72 \\ \cline{1-4} \cline{6-10} \cline{12-16} \cline{18-22}
 & 500 & 50 & 5 &  & 1 & 1 & 1 & 0 & 1 &  & 22.88 & 14.82 & 5.14 & 10.05 & 7.38 &  & 543.02 & 580.47 & 532.92 & 600.00 & 572.86 \\ 
 &  & 71 & 5 &  & 0 & 0 & 0 & 0 & 0 &  & 29.64 & 8.72 & 13.79 & 14.45 & 59.70 &  & 600.00 & 600.00 & 600.00 & 600.00 & 600.00 \\ 
 &  & 125 & 5 &  & 1 & 1 & 1 & 1 & 1 &  & 27.70 & 25.62 & 7.18 & 9.96 & 24.39 &  & 511.39 & 506.27 & 506.38 & 496.78 & 499.03 \\ 
 &  & 166 & 5 &  & 0 & 0 & 0 & 0 & 0 &  & 40.21 & 5.54 & 9.43 & 11.40 & 7.00 &  & 600.00 & 600.00 & 600.00 & 600.00 & 600.00 \\ \cline{1-4} \cline{6-10} \cline{12-16} \cline{18-22}
 & 1000 & 100 & 5 &  & 0 & 0 & 0 & 0 & 0 &  & 35.97 & 36.54 & 21.06 & 47.28 & 80.13 &  & 600.00 & 600.00 & 600.00 & 600.00 & 600.00 \\ 
 &  & 142 & 5 &  & 0 & 0 & 0 & 0 & 0 &  & 40.32 & 42.06 & 35.28 & 24.33 & 63.46 &  & 600.00 & 600.00 & 600.00 & 600.00 & 600.00 \\ 
 &  & 250 & 5 &  & 0 & 0 & 0 & 0 & 0 &  & 42.90 & 6.92 & 11.25 & 20.89 & 15.99 &  & 600.00 & 600.00 & 600.00 & 600.00 & 600.00 \\ 
 &  & 333 & 5 &  & 0 & 0 & 0 & 0 & 0 &  & 46.24 & 23.99 & 38.65 & 22.89 & 12.21 &  & 600.00 & 600.00 & 600.00 & 600.00 & 600.00 \\ \cline{1-4} \cline{6-10} \cline{12-16} \cline{18-22}
 & 2000 & 200 & 5 &  & 0 & 0 & 0 & 0 & 0 &  & 45.23 & 41.22 & 49.48 & 55.55 & 96.77 &  & 600.00 & 600.00 & 600.00 & 600.00 & 600.00 \\ 
 &  & 285 & 5 &  & 0 & 0 & 0 & 0 & 0 &  & 49.21 & 48.07 & 30.01 & 35.97 & 70.28 &  & 600.00 & 600.00 & 600.00 & 600.00 & 600.00 \\ 
 &  & 500 & 5 &  & 0 & 0 & 0 & 0 & 0 &  & 47.16 & 15.84 & 34.96 & 12.70 & 16.99 &  & 600.00 & 600.00 & 600.00 & 600.00 & 600.00 \\ 
 &  & 666 & 5 &  & 0 & 0 & 0 & 0 & 0 &  & 50.10 & 16.19 & 18.53 & 17.62 & 12.34 &  & 600.00 & 600.00 & 600.00 & 600.00 & 600.00 \\ \cline{1-4} \cline{6-10} \cline{12-16} \cline{18-22}
 & 3000 & 300 & 5 &  & 0 & 0 & 0 & 0 & 0 &  & 49.90 & 41.54 & 49.74 & 55.26 & 96.77 &  & 600.00 & 600.00 & 600.00 & 600.00 & 600.00 \\ 
 &  & 428 & 5 &  & 0 & 0 & 0 & 0 & 0 &  & 54.59 & 48.37 & 48.06 & 54.74 & 96.48 &  & 600.00 & 600.00 & 600.00 & 600.00 & 600.00 \\ 
 &  & 750 & 5 &  & 0 & 0 & 0 & 0 & 0 &  & 57.42 & 35.11 & 35.93 & 22.29 & 16.21 &  & 600.00 & 600.00 & 600.00 & 600.00 & 600.00 \\ 
 &  & 1000 & 5 &  & 0 & 0 & 0 & 0 & 0 &  & 58.55 & 8.96 & 9.89 & 9.77 & 9.23 &  & 600.00 & 600.00 & 600.00 & 600.00 & 600.00 \\ \cline{1-4} \cline{6-10} \cline{12-16} \cline{18-22}
\multicolumn{ 3}{c}{avg/sum} & 180 &  & 84 & 84 & 83 & {87} & 81 &  & 21.94 & {12.64} & {12.64} & 12.92 & 20.85 &  & 354.55 & 354.52 & 351.12 & {344.86} & 353.50 \\ \bottomrule
\end{tabular}
}
\label{tab:newinstfloat}
\end{table}

As can be observed in Table \ref{tab:CPCP-results}, SS solves the largest number of instances with integer distances. However, the best average gap is attained by L2. This is because SS cannot improve the upper bound if it is unable to prove optimality. Moreover, BS may spend a significant time solving one or more feasible iterations with an overestimated value of $r$, obtaining a worse overall \#opt value. The layered searches constitute an alternative between SS and BS: progressing faster than a sequential search toward the optimal solution, and avoiding CSCP-$r$ subproblems with large $r$. In terms of CPU time, the global average of each method is similar, because these values are strongly influenced by the time spent on the largest instances, where the time limit is usually reached. Still, significant effects can be observed on some subgroups of instances. For small and medium instances especially, the layered and binary searches usually prove optimality more quickly, because of their reduced number of subproblem iterations.

The results of Table \ref{tab:newinstfloat}, in the presence of floating-point distances, complete the analysis with a different perspective. In this case, the layered searches lead to better results in terms of proven optimal solutions (L4), average gaps (L2 and L3), and time (L4). This is because the number of iterations required by the sequential search is very large in comparison to the other algorithms, and the speedup related to the solution of small subproblems is not sufficient to counterbalance the number of iterations. BS is not a good option either, since it fails to find the optimal solution on a few small instances (e.g., $(n,p) = (50,12)$) or leads to large gaps (e.g., $(n,p) = (1000,100)$ or $(2000,200)$) because of large overestimates of $r$. The layered searches have much more stable performance, with smaller gaps in the wide majority of cases. 

\FloatBarrier
\section{Conclusions} \label{sec:conclusion}

In this paper, we have presented and evaluated different formulations and search algorithms for the capacitated $p$-center problem. We have proposed valid inequalities and symmetry-breaking constraints, as well as an alternative layered search strategy for the distance bound which behaves in between the traditional sequential search and binary search. Our combination of these techniques solved to proven optimality, for the first time, all the instances from the literature {with $n \leq 402$ and integer distances}, including 13 open cases. 
The inequalities, especially the disjunctive-based subset sum inequality, helped to improve the lower bound and to reduce the number of nodes explored. The layered search effectively limits the number of feasible subproblems and avoids large overestimates of the distance. It appears to be especially useful for instances with a large number of distinct distance values, a typical situation in the presence of floating-point distances, where both sequential search and binary search become inefficient.

To stimulate future research, we have also introduced a  new set of larger instances, ranging from 300 to 3000 nodes, with integer or floating-point distances. This set of instances is currently at the frontier of the solution capabilities of exact methods. Future research could focus on exact decomposition methods for the solution of the arc-flow formulation, which has been shown to lead to higher-quality bounds despite its high time consumption, or on the development of additional cutting planes and heuristic callback procedures. Since many applications arising from the machine-learning literature involve problems with thousands or even millions of data points, we also recommend the investigation of aggregation and dominance techniques, to achieve smaller gaps for very large instances. 

\section*{Acknowledgments}
\label{acknowledgments}

This research was partially supported by Conselho Nacional de Desenvolvimento Cient{\'i}fico e Tecnol{\'o}gico (CNPq/Brazil) under Grant \sloppy{GDE 201222/2014-0}, by Coordena{\c{c}}{\~a}o de Aperfei{\c{c}}oamento de Pessoal de N{\'i}vel Superior (CAPES/Brazil) under Grant PVE n. A007/2013, and MIUR/Italy under grant PRIN 2015.

\bibliography{ref}

\end{document}